\documentclass[
  acmsmall, 
  screen, 
  natbib, 
  nonacm, 
]{acmart}

\setcopyright{acmcopyright}
\copyrightyear{20XX} 
\acmYear{20XX} 
\acmDOI{10.1145/XXYYZZZ}

\acmJournal{TOMS}
\acmVolume{1} 
\acmNumber{1} 
\acmArticle{1} 
\acmMonth{1} 

\usepackage{tikz}
\usetikzlibrary{shapes,arrows,positioning,calc,fit,angles,quotes,3d,patterns,fadings,math}

\usepackage{pgfplots}
\usepackage{pgfplotstable}
\pgfplotsset{compat=1.16}
\usepgfplotslibrary{dateplot,groupplots}

\usepackage[utf8]{inputenc} 

\newcommand{\dd}{\mathop{}\text{d}} 
\newcommand{\DD}{\mathop{}\text{D}} 
\newcommand{\pdv}[2]{\frac{\partial#1}{\partial#2}}

\newcommand{\mdv}[2]{\frac{\DD#1}{\DD#2}}

\newcommand{\definedAs}{\coloneqq}

\newcommand{\Knud}{\operatorname{Kn}}

\newcommand{\te}[1]{\boldsymbol{#1}}
\newcommand{\tee}[1]{\te{#1}}
\newcommand{\teee}[1]{\te{#1}}

\usepackage{bm} 

\newcommand{\con}[1]{\ensuremath{\mathrm{#1}}}

\usepackage{mathtools} 

\DeclarePairedDelimiter{\abs}{\lvert}{\rvert}
\DeclarePairedDelimiter{\norm}{\lVert}{\rVert}

\usepackage{amsmath,amsopn}

\DeclareMathOperator{\trace}{tr}

\allowdisplaybreaks{} 

\usepackage{listings}
\definecolor{background_c}{rgb}{.95,.95,.95}
\colorlet{comment_c}{green!50!black}
\colorlet{string_c}{red!60!black}
\colorlet{keyword_c}{magenta!90!black}
\colorlet{identifier_c}{black!90!white}
\colorlet{emph_c}{blue!70!white}
\colorlet{specialfunctions_c}{blue!60!green}
\lstdefinestyle{pythonstyle}{
  language=Python,
  numbers=left,
  numbersep=2mm,
  numberstyle=\tiny\color{gray},
  captionpos=b,
  backgroundcolor=\color{yellow!10},
  basicstyle=\footnotesize \ttfamily,
  deletestring=[s]{r"""}{"""}, 
  morecomment=[s]{r"""}{"""}, 
  morecomment=[s]{"""}{"""}, 
  otherkeywords={\ , \}, \{, \%, \&, \|}, 
  keywordstyle=\color{keyword_c}\bfseries,
  commentstyle=\color{comment_c}\itshape,
  stringstyle=\color{string_c},
  identifierstyle=\color{identifier_c},
  emphstyle=\color{blue}\bfseries,
  emph={[2]True, False, None},
  emphstyle=[2]\color{keyword_c},
  emph={[4]ode, fsolve, sqrt, exp, sin, cos, arccos, pi,  array, norm, solve, dot, arange, isscalar, max, sum, flatten, shape, reshape, find, any, all, abs, plot, linspace, legend, quad, polyval,polyfit, hstack, concatenate,vstack,column_stack,empty,zeros,ones,rand,vander,grid,pcolor,eig,eigs,eigvals,svd,qr,tan,det,logspace,roll,min,mean,cumsum,cumprod,diff,vectorize,lstsq,cla,eye,xlabel,ylabel,squeeze},
  emphstyle=[4]\color{specialfunctions_c},
  emph={[10]jump, grad, tr, transpose, Identity, div, dx, ds, dS, inner, dot, as_tensor, indices, Measure, CellDiameter},
  emphstyle=[10]\color{specialfunctions_c},
  breaklines=false,
  breakatwhitespace=true,
  breakindent=40pt,
  keepspaces=true,
  columns=fullflexible, 
  showtabs=false,
  tab=, 
  tabsize=4,
  showstringspaces=false,
  showspaces=false,
  escapeinside={(*@}{@*)}, 
  extendedchars=\true, 
  xleftmargin=2mm,
  xrightmargin=0mm,
  framexleftmargin=0mm,
  frame=single,
}

\usepackage[capitalize]{cleveref}
\crefname{equation}{}{}

\usepackage{subcaption}

\usepackage{amsthm}
\theoremstyle{remark}
\newtheorem{remark}{Remark}

\begin{document}

\title[fenicsR13: A Tensorial Mixed FEM Solver for the Linear R13 Equations Using FEniCS]{fenicsR13: A Tensorial Mixed Finite Element Solver for the Linear R13 Equations Using the FEniCS Computing Platform
}

\author{Lambert Theisen}
\email{theisen@mathcces.rwth-aachen.de}
\orcid{0000-0001-5460-5425} 
\author{Manuel Torrilhon}
\email{mt@mathcces.rwth-aachen.de}
\orcid{0000-0003-0008-2061} 
\affiliation{%
  \institution{Center for Computational Engineering Science, Department of Mathematics, RWTH Aachen University}
  \streetaddress{Schinkelstr.~2}
  \city{Aachen}
  \country{Germany}
  \postcode{D-52062}
}


\begin{abstract}
  We present a mixed finite element solver for the linearized R13 equations of non-equilibrium gas dynamics. The Python implementation builds upon the software tools provided by the FEniCS computing platform. We describe a new tensorial approach utilizing the extension capabilities of FEniCS's Unified Form Language (UFL) to define required differential operators for tensors above second degree. The presented solver serves as an example for implementing tensorial variational formulations in FEniCS, for which the documentation and literature seem to be very sparse. Using the software abstraction levels provided by the UFL allows an almost one-to-one correspondence between the underlying mathematics and the resulting source code. Test cases support the correctness of the proposed method using validation with exact solutions. To justify the usage of extended gas flow models, we discuss typical application cases involving rarefaction effects. We provide the documented and validated solver publicly.
\end{abstract}




\maketitle


\section{Introduction}
Nowadays, scientists often consider computational resources as the limiting factor in numerical simulations. However, this is not true in general. For non-standard gas flow conditions, the model accuracy can dominate the discretization errors \textendash\ quantitative and qualitative. Especially the lack of essential rarefaction effects in the numerical solution affects the quality of the computational predictions. These non-standard conditions often occur for high Knudsen numbers, e.g., in diluted gases.
For rarefied gas applications, classical Navier--Stokes and Fourier (NSF) models are not sufficiently accurate to predict all occurring non-equilibrium effects. Therefore, we will consider the regularized 13-moment (R13) equations as a natural extension to the classical gas flow models. This set of equations is derived using the moment method for the Boltzmann equation (compare, e.g., with~\cite{grad1958principles}), resulting in additional evolution equations for the heat flux vector and the stress tensor (in contrast to the NSF system). Based on a pseudo-equilibrium approach, regularization terms are added in~\cite{struchtrup2003regularization} to transform \citeauthor{grad1958principles}'s 13-moment equations into the R13 equations.
\par
Several numerical methods were applied to solve the resulting set of equations. In~\cite{rana2013robust}, \citeauthor{rana2013robust} applied a finite difference scheme to obtain steady-state approximations of the R13 equations. Only recently, \citeauthor{torrilhon2017hierarchical} used a discontinuous Galerkin approach in~\cite{torrilhon2017hierarchical} for a hierarchical simulation context. During the same period, \citeauthor{westerkamp2019finite} proposed the first finite element approaches for the steady-state linearized R13 system in~\cite{westerkamp2019finite,westerkamp2017continous} after subsequent advances regarding instability issues and stabilization techniques in~\cite{westerkamp2012finite,westerkamp2014stabilization,westerkamp2017curvature}. Earlier FEM approaches in~\cite{mueller2010computing} already used the FEniCS simulation framework for a simplified set of equations. We will focus on a Galerkin finite element approach and note that previous work did not provide a tensor-based formulation, being very common in the context of mixed Galerkin methods~\cite{auricchio2004mixed}.

\paragraph{The FEniCS Project}
The FEniCS framework~\cite{alnaesEtAl2015fenics,loggMardalWellsEtAl2012} serves as an optimal computing platform for implementing our method. It allows us to avoid the component-wise derivation of scalar variational forms. FEniCS is an LGPLv3-licensed~\cite{web2007lgplv3} collection of interoperable components for automated scientific computing. The main focus is on solving partial differential equations by the finite element method~\cite{alnaesEtAl2015fenics}.
\par
In terms of usability, FEniCS offers both a modern high-level Python interface and a high-performance C++ backend. The whole framework consists of several \textit{building blocks}~\cite{langtangen2016solving}, i.e., flexible and reusable software components allowing the user to write the mathematical models in a very high-level abstraction. The effort to write significant portions of code is shifted away from the user to the developer by the concept of automated code generation. This approach aims to solve the ``disconnect between mathematics and code'' (e.g., a relatively simple Poisson equation \(-\Delta u = f\) vs.~100--10000 lines of code)~\cite{logg2015ytImplementing}. In the optimal case, a user should only write the model problem's mathematical statements while the simulation framework executes all the extra work automatically.
\par
The main component of FEniCS is DOLFIN~\cite{logg2010dolfin,hoffman2002dolfin}. This library implements the high-performance C++ backend, consisting of the relevant data structures such as meshes, function spaces, and function expressions. The main algorithms in the context of the finite element method (i.e., element assembly, mesh handling, or connection to linear algebra solvers) are also part of DOLFIN\@. From a user perspective, DOLFIN is the main connection layer between the high-level Python interface and the core layer because it handles the communication between all modules and extends them with external software. The most important internal low-level components of FEniCS are:
\begin{itemize}
  \item The UFL~\cite{alnaes2014unified} (Unified Form Language) is a domain-specific language to formulate finite element problems in their weak variational form. With UFL, the user can express a model problem in its natural and mathematical way. It also serves as the input for the form compiler used by both the C++ and the Python interface.
  \item The FFC~\cite{kirby2006compiler} (FEniCS Form Compiler) automatically generates DOLFIN code from a given variational form. The goal of FFC is to provide a validated compiler, performing automated code generation tasks to improve the correctness of the resulting code.
  \item The FIAT~\cite{kirby2004algorithm} (FInite element Automatic Tabulator) enables the automatic computation of basis functions for nearly arbitrary finite elements.
\end{itemize}
Other simulation frameworks building upon or utilizing similar concepts as FEniCS are the Firedrake project~\cite{rathgeber2017firedrake} or FEniCS-HPC~\cite{hoffman2015fenics}. The presented work uses the FEniCS version 2019.1.0.r3.

\paragraph{Outline}
This work's organization is as follows: In \cref{s_modelFormulation}, we present the tensorial R13 model equations, discretized in \cref{s_galerkinFiniteElementApproach} using a Galerkin approach. \cref{s_implementationAndValidation,s_validation} are devoted to implementing and validating the proposed method, focusing on auxiliary implementations for tensor differential operators. Furthermore, application cases in \cref{s_applications} present the solver capabilities to predict the critical flow phenomena. Finally, we discuss limitations and future work before adding some concluding remarks in \cref{s_conclusion}.

\section{Formulation of the Model Equations}\label{s_modelFormulation}
We first consider the general case of a closed and time-independent gas domain \(\tilde{\Omega} \subset \mathbb{R}^{3}\) in three spatial dimensions. The main quantities are the time evolutions (with \(t \in [0,\con{T}]\)) of field quantities, such as the gas density \(\rho : \tilde{\Omega} \times [0,\con{T}] \rightarrow \mathbb{R}_+\), the gas velocity \(\te{u} : \tilde{\Omega} \times [0,\con{T}] \rightarrow \mathbb{R}^d\), and the gas temperature \(T : \tilde{\Omega} \times [0,\con{T}] \rightarrow \mathbb{R}_+\) inside the domain \(\tilde{\Omega}\). These are the three fundamental quantities. We will further encounter the pressure \(p : \tilde{\Omega} \times [0,\con{T}] \rightarrow \mathbb{R}_+\), the heat flux \(\te{s} : \tilde{\Omega} \times [0,\con{T}] \rightarrow \mathbb{R}^d\), and the deviatoric stress tensor \(\tee{\sigma} : \tilde{\Omega} \times [0,\con{T}] \rightarrow \mathbb{R}_{\text{stf}}^{d \times d}\). The deviatoric part means the symmetric and trace-free part of a tensor. Note that the traditional symbol \(\te{q}\) for the heat flux is replaced by \(\te{s}\) to have an intuitive set of test functions in the weak formulations later on in \cref{s_galerkinFiniteElementApproach}.
The overall goal is to determine these evolutions for all points \(\te{x} \in \tilde{\Omega}\) from given initial conditions (e.g.,~\(\rho_0(\te{x},0), \te{u}_0(\te{x},0), T(\te{x},0), \ldots\)) together with a set of given boundary conditions. The latter describes the outer environment and boundary behavior of the fields on the domain boundary \(\partial \tilde{\Omega}\).

\subsection{Modeling Non-Equilibrium Gas Flows Using the R13 Equations}
Three fundamental laws of physics describe the behavior of general gas flows: The point-wise conservation of mass in~\cref{eq_consMass}\index{balance laws!mass}, the balance of momentum in~\cref{eq_consMomentum}\index{balance laws!momentum} (Newton's second law of motion), and the energy conservation in~\cref{eq_consEnergy}\index{balance laws!energy} (first law of thermodynamics), expressed in component form as
\begin{subequations}\label{eqsec:balanceLawsLagrangian}
  \begin{align}
    \mdv{\rho}{t} + \rho \pdv{u_k}{x_k} &= \dot{m}, \label{eq_consMass} \\
    \rho \mdv{u_i}{t} + \pdv{p}{x_i} + \pdv{\sigma_{ij}}{x_j} &= \rho b_i, \label{eq_consMomentum} \\
    \rho \mdv{\epsilon}{t} + p \pdv{u_i}{x_i} + \sigma_{ij} \pdv{u_i}{x_j} + \pdv{s_i}{x_i} &= r. \label{eq_consEnergy}
  \end{align}
\end{subequations}
In~\eqref{eqsec:balanceLawsLagrangian}, \(\{\dot{m}, b_i, r\}\) denotes a mass source, a body force, and an energy source, respectively. \( (\text{D}\star/\text{D}t) = (\partial\star/\partial t) + u_j (\partial\star/\partial x_j)\) defines the material derivative using Einstein's summation convention \textendash\ compare, e.g., with~\cite{schade2009tensoranalysis}. When assuming a monoatomic ideal gas, the pressure is related to the density and the temperature as \(p=\rho \theta\). In contrast to \(T\), \(\theta\) is the temperature in energy units as \(\theta = (k/m) T\). The classical constitutive theory considers Fourier's law (\(s_i=-\kappa (\partial T)/(\partial x_i)\)) and the law of Navier--Stokes (\(\sigma_{ij} = -2 \mu (\partial u_{\langle i})/(\partial x_{j\rangle})\) in notation~\eqref{eq_rank2stf}) as closure relations to solve the system~\eqref{eqsec:balanceLawsLagrangian} \textendash\ altogether forming the well-known NSF system for compressible gas flows.
\par
However, the NSF model is only accurate near the thermodynamic equilibrium with \(\Knud \ll 1\), where the dimensionless Knudsen number
\begin{equation}
  \Knud=\frac{\lambda}{L},
\end{equation}
describes the ratio between the mean free path \(\lambda\) of the gas particles and the relevant characteristic length scale \(L\)~\cite{struchtrup2005macroscopic}. One observes many interesting rarefaction effects or micro-scale phenomena in a non-equilibrium gas flow at about \(\Knud \gtrsim 0.05 \). Experiments and the underlying Boltzmann equation's analysis confirmed these observations~\cite{struchtrup2011macroscopic}. We find a comprehensive list of rarefaction effects in~\cite{struchtrup2011macroscopic,torrilhon2016modeling}. Some of them are:
\begin{itemize}
  \item Heat flux parallel to the walls in a channel flow, contradicting Fourier's law as there is no temperature difference in this direction (see \cref{sec_channelFlow});
  \item A non-constant pressure behavior in Couette and Poiseuille channel flows although no flow across the channel is present;
  \item A minimum of the mass flow rate (\textit{Knudsen minimum}) in a force-driven Poiseuille flow, also known as the \textit{Knudsen paradox} (see \cref{sec_channelFlow});
  \item A non-convex temperature profile in such microchannels while NSF predicts a strictly convex shape for the same setup;
  \item Temperature-induced flow situations in channels (see \cref{sec_knudsenPump});
  \item Temperature jump and velocity slip at walls (\textit{Knudsen boundary layers});
  \item as well as Temperature-induced edge flow (see \cref{sec_thermalEdgeFlow}).
\end{itemize}
\par
Therefore, the goal of a rarefied gas solver is to predict all of the above-listed effects accurately. Using Boltzmann's transport equation is an option for all flow situations with \(\Knud \in \mathbb{R}_+\). However, due to high dimensionality, its numerical simulation is costly compared to classical continuum approaches. For this reason, there exists a variety of models that extend the classical NSF system. We use the R13 equations proposed in~\cite{struchtrup2003regularization} (summarized in~\cite{struchtrup2011macroscopic,torrilhon2016modeling}). This prominent example of extended macroscopic gas flow models is suitable in the transition regime away from the thermodynamic equilibrium. Accordingly, one might consider the full set of non-linear R13 equations as the natural extension of the Navier--Stokes' and Fourier's models to more field equations.

\subsection{Steady-State Linearization and Simplifications}\label{ss_linearization}
Throughout this work, we will consider the steady-state and linearized R13 equations using pressure \(p\) and temperature \(\theta\) instead of density \(\rho\) and classical temperature \(T\). This set of variables is a common choice for engineering applications. The full set of non-linear equations consists of the three conservation equations \cref{eq_consMass,eq_consMomentum,eq_consEnergy} and balance laws for \(\te{s}\) and \(\tee{\sigma}\) and are given, e.g., in~\cite{struchtrup2003regularization,struchtrup2011macroscopic,torrilhon2016modeling}. The balance laws for \(\te{s}\) and \(\tee{\sigma}\) contain collision terms with the collision frequency \(\nu\). A reformulation expands the material derivatives. It uses \(p=\rho \theta\) and \(\epsilon = 3 \theta/2\), accounting for monoatomic ideal gases. We further neglect all temporal operators and only consider deviation fields \(\delta \te{U}\) with \(\te{U}=\te{U}_0 + \delta \te{U}\) around a ground state \(\te{U}_0 = (\rho^{(0)},{u_i}^{(0)},\theta^{(0)},{\sigma_{ij}}^{(0)},{s_i}^{(0)}) = (\rho_0,0,\theta_0,0,0)\). This procedure linearizes the whole system together with the closures.
\par
We end up with the only remaining parameter \(\tau\), which is the mean free-flight time defined as \(\tau = 1/\nu\). A subsequent scaling of the equations relates every field \(\star\) to its reference \(\hat{\star}\). We therefore define
\begin{equation}
  \hat{x}_i \definedAs \frac{x_i}{L},
  \;
  \hat{\theta} \definedAs \frac{\theta}{\theta_0},
  \;
  \hat{p} \definedAs \frac{p}{p_0},
  \;
  \hat{\tau} \definedAs \frac{\tau}{\tau_0},
  \;
  \hat{u}_i \definedAs \frac{u_i}{\sqrt{\theta_0}},
\end{equation}
which leads to the references
\begin{equation}
  \hat{\sigma}_{ij} = \frac{\sigma}{p_0},
  \;
  \hat{s}_i = \frac{s_i}{p_0\sqrt{\theta_0}},
  \;
  \hat{m}_{ijk} = \frac{m_{ijk}}{p_0\sqrt{\theta_0}},
  \;
  \hat{\Delta}_{ij} = \frac{\Delta_{ij}}{p_0 \theta_0}.
  \;
\end{equation}
We then replace all quantities with their dimensionless counterpart multiplied with the reference value. For example, we insert \(x_i \rightarrow L \hat{x}_i\) into the linearized model. The resulting equations, then, allow for the identification of the Knudsen number as
\begin{equation}
  \Knud = \frac{\tau_0 \sqrt{\theta_0}}{L}.
\end{equation}
The resulting system of interest is linear, steady-state, and dimensionless. We switch to a tensorial notation and drop the dimensionless indicator \(\hat{\star}\) for better readability of the differential operators and simpler notation. The resulting balance laws read
\begin{align}
  \te{\nabla} \te{\cdot} \te{u} &= \dot{m},\label{eq_balance_mass}
  \\
  \te{\nabla} p + \te{\nabla} \te{\cdot} \tee{\sigma} &= \te{b},\label{eq_balance_momentum}
  \\
  \te{\nabla} \te{\cdot} \te{u} + \te{\nabla} \te{\cdot} \te{s} &= r,\label{eq_balance_energy}
\end{align}
with additional evolution equations for the heat flux vector \(\te{s}\) and the deviatoric stress tensor \(\tee{\sigma}\) as
\begin{align}
  \frac{4}{5} {(\te{\nabla} \te{s})}_{\text{stf}} + 2 {(\te{\nabla} \te{u})}_{\text{stf}} + \te{\nabla} \te{\cdot} \teee{m} &= - \frac{1}{\Knud} \tee{\sigma},\label{eq_balance_heatflux}
  \\
  \frac{5}{2} \te{\nabla} \theta + \te{\nabla} \te{\cdot} \tee{\sigma} + \frac{1}{2} \te{\nabla} \te{\cdot} \tee{R} + \frac{1}{6} \te{\nabla} \Delta &= - \frac{1}{\Knud} \frac{2}{3} \te{s}.\label{eq_balance_stress}
\end{align}
To obtain a closed system, we also require the linearized closure relations for the highest-order moments as
\begin{align}
  \teee{m} &= - 2 \Knud {(\te{\nabla} \tee{\sigma})}_{\text{stf}},\label{eq_closure_m}
  \\
  \tee{R} &= - \frac{24}{5} \Knud {(\te{\nabla} \te{s})}_{\text{stf}},\label{eq_closure_r}
  \\
  \Delta &= - 12 \Knud \left( \te{\nabla} \te{\cdot} \te{s} \right).\label{eq_closure_delta}
\end{align}
Here, the symmetric and trace-free (deviatoric) part of a 2-tensor is defined component-wise \({(\star)}_{ij} \mapsto {({(\star)}_{\text{stf}})}_{ij} = {(\star)}_{\langle ij \rangle}\), with the trace subtracted from the symmetric part \({(\star)}_{(ij)}\). For a tensor \(\tee{A} \in \mathbb{R}^{3 \times 3}\) and following~\cite{struchtrup2005macroscopic}, this translates to
\begin{equation}\label{eq_rank2stf}
  A_{\langle ij \rangle} = A_{(ij)} - \frac{1}{3} A_{kk} \delta_{ij} = \frac{1}{2} (A_{ij} + A_{ji}) - \frac{1}{3} A_{kk} \delta_{ij},
\end{equation}
using Kronecker's delta function \(\delta_{ij}\). The symmetric and trace-free part~\cite{struchtrup2005macroscopic} of a 3-tensor \(\teee{B} \in \mathbb{R}^{3 \times 3 \times 3}\) analogously reads
\begin{equation}\label{eq_rank3stf}
  B_{\langle ijk \rangle} = B_{(ijk)} - \frac{1}{5} \left( B_{(ill)} \delta_{jk} + B_{(ljl)} \delta_{ik} + B_{(llk)} \delta_{ij}\right).
\end{equation}
Here, the symmetric part of a 3-tensor is the average of all possible transpositions and is given by
\begin{equation}\label{eq_rank3sym}
  B_{(ijk)} = \frac{1}{6} \left( B_{ijk} + B_{ikj} + B_{jik} + B_{jki} + B_{kij} + B_{kji} \right).
\end{equation}
\par
We restrict the set of computational domains \(\tilde{\Omega} \subset \mathbb{R}^3\) to geometries \(\tilde{\Omega} \equiv \Omega \times \mathbb{R}\) with complete homogeneity in the third spatial direction \(z\), such that \(\partial_{x_3} = \partial_z \equiv 0\). This assumption simplifies the calculation of all variables, although they formally remain three-dimensional. Therefore, in this paper's remainder, we only have to consider the computational domain \(\Omega \subset \mathbb{R}^2\). One can assume this as an area cut from an infinitely long domain in the \(z\)-direction.

\subsection{Linearized Boundary Conditions}
To formulate boundary value problems, we require a set of linearized boundary conditions. In~\cite{rana2016thermodynamically}, \citeauthor{rana2016thermodynamically} proposed the most recent version based on Maxwell's accommodation model~\cite{torrilhon2008boundary} while we use the notation of~\cite{westerkamp2019finite,torrilhon2017hierarchical} as
\begin{align}
  u_n &= 0,\label{eq_bc_un}
  \\
  \sigma_{nt} &= \tilde{\chi} \left( (u_t-u_t^{\mathrm{w}}) + \frac{1}{5} s_t + m_{nnt} \right),\label{eq_bc_sigmant}
  \\
  R_{nt} &= \tilde{\chi} \left( -(u_t-u_t^{\mathrm{w}}) + \frac{11}{5} s_t - m_{nnt} \right),\label{eq_bc_rnt}
  \\
  s_n &= \tilde{\chi} \left( 2(\theta-\theta^{\mathrm{w}}) + \frac{1}{2} \sigma_{nn} + \frac{2}{5} R_{nn} + \frac{2}{15} \Delta \right),\label{eq_bc_sn}
  \\
  m_{nnn} &= \tilde{\chi} \left( - \frac{2}{5} (\theta-\theta^{\mathrm{w}}) + \frac{7}{5} \sigma_{nn} - \frac{2}{25} R_{nn} - \frac{2}{75} \Delta \right),\label{eq_bc_mnnn}
  \\
  \left( \frac{1}{2} m_{nnn} + m_{ntt} \right) &= \tilde{\chi} \left( \frac{1}{2} \sigma_{nn} + \sigma_{tt} \right).\label{eq_bc_05mnnnmnnt}
\end{align}
A two-dimensional local boundary-aligned coordinate system in terms of outer normal and tangential components \((\te{n},\te{t})\) generates the required projections. The modified accommodation factor is given by \(\tilde{\chi} = \sqrt{2/(\pi \theta_0)} \chi/(2-\chi)\)~\cite{westerkamp2019finite}. The boundary conditions are equal to the Onsager boundary conditions of~\cite{rana2016thermodynamically}. They were adjusted, as described in~\cite{torrilhon2017hierarchical} (to ensure thermodynamic admissibility).
\par
For real-life applications, it is often necessary to prescribe inflow or outflow conditions. The trivial velocity boundary condition in the \(\te{n}\)-direction (\(u_n=0\)) is therefore replaced by an inflow model, following the idea of~\cite{torrilhon2017hierarchical} as
\begin{equation}
  \epsilon^{\mathrm{w}} \tilde{\chi} \left( (p-p^{\mathrm{w}}) + \sigma_{nn} \right) = \left(u_n - u_n^{\mathrm{w}}\right). \label{eq_inflowBC}
\end{equation}
The artificial in- and outflow interface require a pressure \(p^{\mathrm{w}}\), a velocity \(u_n^{\mathrm{w}}\), and a velocity prescription coefficient \(\epsilon^\mathrm{w}\). Although these interfaces are no physical walls, we still use the notation \(\star^{\textrm{w}}\) for indicating the boundary values following the other conditions \cref{eq_bc_sigmant,eq_bc_rnt,eq_bc_sn,eq_bc_mnnn,eq_bc_05mnnnmnnt}. Intuitively, the equation \(\epsilon^\mathrm{w}=0\), together with \(u_n^{\mathrm{w}}=0\), reduces the inflow model back to the standard boundary condition. However, a value of \(\epsilon^\mathrm{w} \rightarrow \infty\) enforces the total pressure at the wall \(p^{\mathrm{w}} = p + \sigma_{nn}\).

\section{Galerkin Finite Element Approach}\label{s_galerkinFiniteElementApproach}
The R13 equations of \cref{ss_linearization} are solved numerically. Before utilizing the Galerkin finite element method in \cref{ss_femDiscretization}, \cref{ss_weakform} presents the weak formulation. A stabilization approach using a continuous interior penalty (CIP) method is proposed in the \cref{ss_stabilization} and allows for a broader spectrum of stable element combinations.

\subsection{Derivation of the Variational Formulation}\label{ss_weakform}
The derivation of the weak variational form follows the usual strategy. The first step consists of integration over the computational domain \(\Omega\) while multiplying (testing) with a corresponding set of test functions. Due to the different tensorial degrees of all five equations, the trial- and test-function vectors read
\begin{equation}
  \te{\mathcal{U}} \definedAs (\te{s},\theta,\tee{\sigma},\te{u},p),
  \;
  \te{\mathcal{V}} \definedAs (\te{r},\kappa,\tee{\psi},\te{v},q).
\end{equation}
One 2-tensorial, two vectorial, and two scalar test functions from suitable Sobolev function spaces \(\mathbb{V}_\star\). We choose the trial and test functions from the same product spaces for the Galerkin method as \(\te{\mathcal{U}},\te{\mathcal{V}} \in \mathbb{H} \definedAs \bigtimes_{i \in \te{\mathcal{U}}} \mathbb{V}_i = \mathbb{V}_{\te{s}} \times \mathbb{V}_{\theta} \times \mathbb{V}_{\tee{\sigma}} \times \mathbb{V}_{\te{u}} \times \mathbb{V}_{p}\). In \cref{ss_derviationHeatflux,ss_derviationEnergy,ss_derviationStress,ss_derviationMomentum,ss_derviationMass}, we transform all evolution equations using the following strategy:
\begin{enumerate}
  \item Integrate over \(\Omega\) while testing with the corresponding test function \(\star \in \mathcal{V}\).
  \item Apply integration by parts to all \(\te{u}\)-, \(\te{r}\)-, and \(\tee{m}\)-terms.
  \item Insert the corresponding conditions on the boundary \(\Gamma \definedAs \partial \Omega\) using a normal/tangential aligned coordinate system with \((\te{n},\te{t})\).
  \item Use the closure relations \cref{eq_closure_m,eq_closure_r,eq_closure_delta} to eliminate the highest-order moments \(\teee{m},\tee{R}\), and \(\Delta\).
\end{enumerate}
The addition of all five weak equations \cref{eq_subf_line_heatflux,eq_subf_line_energy,eq_subf_line_stress,eq_subf_line_momentum,eq_subf_line_mass} yields the continuous compound formulation: \textit{Find} \(\te{\mathcal{U}} \in \mathbb{H}\) \textit{, such that} for all \(\te{\mathcal{V}} \in \mathbb{H}\):
\begin{equation}
  \mathcal{A}
  \left(\te{\mathcal{U}},\te{\mathcal{V}}\right)
  =
  \mathcal{L}\left(\te{\mathcal{V}}\right)
  .
\end{equation}
A collection and structuring step produces the bilinear form \(\mathcal{A}: \mathbb{H} \times \mathbb{H} \rightarrow \mathbb{R}\) on the product space over \(\mathbb{H}\). By defining the sub-functionals \(a(\te{s},\te{r}),\ldots,h(p,q)\), the combined weak form reads
\begin{align}
  \mathcal{A}\left(\te{\mathcal{U}},\te{\mathcal{V}}\right)
  &=
  a(\te{s},\te{r})
  +
  b(\kappa, \te{s})
  -
  b(\theta, \te{r})
  +
  c(\te{s},\tee{\psi})
  -
  c(\te{r},\tee{\sigma})
  +
  d(\tee{\sigma},\tee{\psi})
  \nonumber
  \\
  &+
  e(\te{u},\tee{\psi})
  -
  e(\te{v},\tee{\sigma})
  +
  f(p,\tee{\psi})
  +
  f(q,\tee{\sigma})
  +
  g(p,\te{v})
  -
  g(q,\te{u})
  +
  h(p,q)
  ,
  \label{eq_compundWeakForm}
\end{align}
while the linear functional \(\mathcal{L}: \mathbb{H} \rightarrow \mathbb{R}\) on the right-hand side reads
\begin{align}
  \mathcal{L}\left(\mathcal{V}\right)
  &=
  l_1(\te{r}) +
  l_2(\kappa) +
  l_3(\tee{\psi}) +
  l_4(\te{v}) +
  l_5(q)
  .
  \label{eq_compoundRHS}
\end{align}
\par
The bilinear sub-functionals used in~\eqref{eq_compundWeakForm} contain \(a(\te{s},\te{r})\), \(d(\tee{\sigma},\tee{\psi})\), \(h(p,q)\) as symmetric diagonal terms. Considering the physical interpretations, \(b(\star,\star)\) is an intra-heat coupling, \(e(\star,\star)\), \(f(\star,\star)\), \(g(\star,\star)\) are intra-stress couplings, and the contribution \(c(\star,\star)\) is the inter-heat-stress coupling. Altogether, they read:
\begin{align}
  a(\te{s},\te{r})
  &=
  \frac{24}{25} \Knud \int_\Omega \text{sym}(\te{\nabla}\te{s}) \tee{:} \text{sym}(\te{\nabla}\te{r}) \dd \te{x}
  +
  \frac{12}{25} \Knud \int_\Omega \text{div}(\te{s}) \text{div}(\te{r}) \dd \te{x}
  \nonumber
  \\
  &
  +
  \frac{4}{15} \frac{1}{\Knud} \int_\Omega \te{s} \cdot \te{r} \dd \te{x}
  + \frac{1}{2} \frac{1}{\tilde{\chi}} \int_\Gamma s_n r_n \dd l
  + \frac{12}{25} \tilde{\chi} \int_\Gamma s_t r_t \dd l
  \label{eq_subf_a}
  ,
  \\
  b(\theta, \te{r})
  &=
  \int_\Omega \theta \, \text{div}(\te{r}) \dd \te{x}
  \label{eq_subf_b}
  ,
  \\
  c(\te{r},\tee{\sigma})
  &=
  \frac{2}{5} \int_\Omega \tee{\sigma} \tee{:} \te{\nabla} \te{r} \dd \te{x}
  - \frac{3}{20} \int_\Gamma \sigma_{nn} r_n \dd l
  - \frac{1}{5} \int_\Gamma \sigma_{nt} r_t \dd l
  \label{eq_subf_c}
  ,
  \\
  d(\tee{\sigma},\tee{\psi})
  &=
  \Knud \int_\Omega \text{stf}(\te{\nabla}\tee{\sigma}) \teee{\because} \text{stf}(\te{\nabla}\tee{\psi}) \dd \te{x}
  +
  \frac{1}{2} \frac{1}{\Knud} \int_\Omega \tee{\sigma} \tee{:} \tee{\psi} \dd \te{x}
  \nonumber
  \\
  &+
  \frac{9}{8} \tilde{\chi} \int_\Gamma \sigma_{nn} \psi_{nn} \dd l
  +
  \tilde{\chi} \int_\Gamma \left( \sigma_{tt} + \frac{1}{2} \sigma_{nn} \right) \left( \psi_{tt} + \frac{1}{2} \psi_{nn} \right) \dd l
  \nonumber
  \\
  &+
  \frac{1}{\tilde{\chi}} \int_\Gamma \sigma_{nt} \psi_{nt} \dd l
  +
  \epsilon^{\text{w}} \tilde{\chi} \int_\Gamma \sigma_{nn} \psi_{nn} \dd l
  \label{eq_subf_d}
  ,
  \\
  e(\te{u},\tee{\psi})
  &=
  \int_\Omega \text{div}(\tee{\psi}) \te{\cdot} \te{u} \dd \te{x}
  \label{eq_subf_e}
  ,
  \\
  f(p,\tee{\psi})
  &=
  \epsilon^{\text{w}} \tilde{\chi} \int_\Gamma p \psi_{nn} \dd l
  \label{eq_subf_f}
  ,
  \\
  g(p,\te{v})
  &=
  \int_\Omega \te{v} \te{\cdot} \te{\nabla} p \dd \te{x}
  \label{eq_subf_g}
  ,
  \\
  h(p,q)
  &=
  \epsilon^{\text{w}} \tilde{\chi}  \int_\Gamma p q \dd l
  \label{eq_subf_h}
  .
\end{align}
The linear functionals of~\eqref{eq_compoundRHS} contain the corresponding source terms, forces, and given boundary expressions:
\begin{align}
  l_1(\te{r}) &= - \int_\Gamma \theta^{\text{w}} r_n \dd l \label{eq_subf_l1}
  ,
  \\
  l_2(\kappa) &= \int_\Omega \left( r - \dot{m} \right) \kappa \dd \te{x} \label{eq_subf_l2}
  ,
  \\
  l_3(\tee{\psi}) &= - \int_\Gamma \left( u_t^{\text{w}} \psi_{nt} + \left( u_n^{\text{w}} - \epsilon^{\text{w}} \tilde{\chi} p^{\text{w}} \right) \psi_{nn} \right) \dd l \label{eq_subf_l3}
  ,
  \\
  l_4(\te{v}) &= \int_\Omega \te{b} \te{\cdot} \te{v} \dd \te{x} \label{eq_subf_l4}
  ,
  \\
  l_5(q) &= \int_\Omega \dot{m} q \dd \te{x} - \int_\Gamma \left( u_n^{\text{w}} - \epsilon^{\text{w}} \tilde{\chi}  p^{\text{w}} \right) q \dd l \label{eq_subf_l5}
  .
\end{align}
Identification of the total pressure \(p + \sigma_{nn}\) allows for an alternative notation of \(\mathcal{A}(\te{\mathcal{U}},\te{\mathcal{V}})\) in \cref{eq_compundWeakForm} as
\begin{align}
  \mathcal{A}\left(\te{\mathcal{U}},\te{\mathcal{V}}\right)
  &=
  a(\te{s},\te{r})
  +
  b(\kappa, \te{s})
  -
  b(\theta, \te{r})
  +
  c(\te{s},\tee{\psi})
  -
  c(\te{r},\tee{\sigma})
  +
  \bar{d}((\tee{\sigma}, p), (\tee{\psi}, q))
  \nonumber
  \\
  &+
  e(\te{u},\tee{\psi})
  -
  e(\te{v},\tee{\sigma})
  +
  g(p,\te{v})
  -
  g(q,\te{u})
  ,
  \label{eq_compundWeakFormAlternative}
\end{align}
where we add the terms \(d(\tee{\sigma},\tee{\psi})\), \(f(p,\tee{\psi})\), \(f(q,\tee{\sigma})\), and \(h(p,q)\) to form \(\bar{d} : (\mathbb{V}_{\tee{\sigma}} \times \mathbb{V}_{p}) \times (\mathbb{V}_{\tee{\sigma}} \times \mathbb{V}_{p}) \rightarrow \mathbb{R}\) as
\begin{align}
  \bar{d}((\tee{\sigma}, p), (\tee{\psi}, q))
  &=
  \Knud \int_\Omega \text{stf}(\te{\nabla}\tee{\sigma}) \teee{\because} \text{stf}(\te{\nabla}\tee{\psi}) \dd \te{x}
  +
  \frac{1}{2} \frac{1}{\Knud} \int_\Omega \tee{\sigma} \tee{:} \tee{\psi} \dd \te{x}
  \nonumber
  \\
  &+
  \frac{9}{8} \tilde{\chi} \int_\Gamma \sigma_{nn} \psi_{nn} \dd l
  +
  \tilde{\chi} \int_\Gamma \left( \sigma_{tt} + \frac{1}{2} \sigma_{nn} \right) \left( \psi_{tt} + \frac{1}{2} \psi_{nn} \right) \dd l
  \nonumber
  \\
  &+
  \frac{1}{\tilde{\chi}} \int_\Gamma \sigma_{nt} \psi_{nt} \dd l
  +
  \epsilon^{\text{w}} \tilde{\chi} \int_\Gamma (p + \sigma_{nn}) (q + \psi_{nn}) \dd l
  \label{eq_subf_dTilde}
  .
\end{align}
The total pressure term \(\epsilon^{\text{w}} \tilde{\chi} \int_\Gamma (p + \sigma_{nn}) (q + \psi_{nn}) \dd l\) replaces the last term \(\epsilon^{\text{w}} \tilde{\chi} \int_\Gamma \sigma_{nn} \psi_{nn} \dd l\) of \(d(\sigma,\psi)\) in \cref{eq_subf_d}. We will need the notation \cref{eq_compundWeakFormAlternative} in \cref{thm_positiveDefiniteness} to simplify notation. The actual implementation, however, uses the equivalent weak form \cref{eq_compundWeakForm} to set up the system \cref{eq_discreteSystem}.

\subsection{Finite Element Discretization}\label{ss_femDiscretization}
We consider a conforming and shape-regular partition \(\mathcal{T}_h\) of the computational domain \(\Omega \subset \mathbb{R}^2\) into triangular elements \(\tau\) as \(\mathcal{T}_h = {\left\{ \tau \right\}}_{\tau \in \mathcal{T}_h}\). For a given polynomial degree \(m \in \mathbb{N}\), let the space of all polynomials with maximal degree \(m\) on every \(\tau \in \mathcal{T}_h\) reads
\begin{equation}
  \mathbb{V}_{\star,h} = \left\{ u \in \mathbb{V}_{\star} : u|_\tau \in \mathbb{P}_m(\tau) \, \forall \tau \in \mathcal{T}_h \right\}.
\end{equation}
Following the usual conforming finite element approach, we restrict the function space to a finite-dimensional subspace \(\mathbb{H}_h \subset \mathbb{H}\) by choosing polynomial ansatz functions for all fields. This discretization procedure, then, leads to the discrete algebraic system:
\newcommand{\tmpstreatch}{1.2}
\begin{equation}
  \renewcommand{\arraystretch}{\tmpstreatch}
  \left[{
    \begin{array}{cc|ccc}
      A_h & -B_h^T & -C_h^T &  0 &  0 \\
      B_h & 0 &  0 &  0 &  0 \\
      \hline
      C_h &  0 &  D_h & -E_h^T &  F_h^T \\
      0 &  0 & E_h & 0 &  G_h^T \\
      0 &  0 & F_h & -G_h & H_h \\
    \end{array}
  }\right]
  \left[{
    \begin{array}{c}
      \te{s}_h \\
      \theta_h \\
      \hline
      \tee{\sigma}_h \\
      \te{u}_h \\
      p_h \\
    \end{array}
  }\right]
  =
  \left[{
    \begin{array}{c}
      L_{1,h} \\
      L_{2,h} \\
      \hline
      L_{3,h} \\
      L_{4,h} \\
      L_{5,h} \\
    \end{array}
  }\right]
  ,
  \label{eq_discreteSystem}
\end{equation}
where each matrix \(A_h,\ldots,H_h\) corresponds to its corresponding weak form \(a(\te{s},\te{r}),\ldots,h(p,q)\). The system's formulation \cref{eq_discreteSystem} reveals the physical coupling between the heat variables \((\te{s}_h,\theta_h)\) and the stress variables \((\tee{\sigma}_h,\te{u}_h,p_h)\) only through the \(C_h\) and \(C_h^T\) matrices.
\par
To show the need for stabilization of \cref{eq_discreteSystem}, we reorder the rows with \(\te{x}={(\tee{\sigma}_h ,\te{s}_h ,p_h)}^T\) and \(\te{y}={(\te{u}_h ,\theta_h)}^T\) such that
\begin{equation}
  \left[{
    \begin{array}{cc}
      \mathbb{A} & -\mathbb{B}^T \\
      \mathbb{B} & \tee{0} \\
    \end{array}
  }\right]
  \left[{
    \begin{array}{c}
      \te{x} \\
      \te{y}
    \end{array}
  }\right]
  =
  \left[{
    \begin{array}{c}
      \te{f} \\
      \te{g}
    \end{array}
  }\right]
  ,
  \label{eq_saddleSystem}
\end{equation}
with
\begin{equation}
  \mathbb{A}
  =
  \left[{
    \begin{array}{ccc}
      D_h & C_h & F_h^T \\
      -C_h^T & A_h & 0 \\
      F_h & 0 & H_h
    \end{array}
  }\right]
  ,
  \;
  \mathbb{B}
  =
  \left[{
    \begin{array}{ccc}
      E_h & 0 & G_h^T \\
      0 & B_h & 0 \\
    \end{array}
  }\right]
  ,
  \;
  \te{f}
  =
  \left[{
    \begin{array}{c}
      L_{3,h} \\
      L_{1,h} \\
      L_{5,h} \\
    \end{array}
  }\right]
  ,
  \;
  \te{g}
  =
  \left[{
    \begin{array}{c}
      L_{4,h} \\
      L_{2,h} \\
    \end{array}
  }\right]
  .
\end{equation}
The notation used in \cref{eq_saddleSystem} reveals the saddle point structure (compare, e.g., with~\cite{auricchio2004mixed}). We directly observe the need for \((\te{u},\theta)\)-stabilization due to the zero diagonal entries. For an impermeable wall condition \(u_n = 0\) resulting from \(\epsilon^{\text{w}}=0\), the \(H_h\)-block is also vanishing. The \(p\)-diagonal, therefore, also needs stabilization to work for all possible boundary conditions.

\subsection{Continuous Interior Penalty (CIP) Stabilization}\label{ss_stabilization}
In general, mixed finite element problems require a compatible set of finite elements. For example, in Stokes's second problem, a suitable choice to circumvent the LBB condition is the Taylor--Hood element \(\mathbb{P}_2\mathbb{P}_1\). Here, the velocity function space has a higher dimension than the pressure function space. When it comes to application cases, we do not want to focus on a particular field and desire an equal order discretization. Especially for higher-order moments, this is true due to no real physical intuition about these fields. The argument gets stronger when considering even more complex models above the 13 field case.
\par
One approach to overcome the compatible condition on the discrete function spaces is stabilization. In general, stabilization techniques modify the weak form's left-hand side to stabilize the discrete system, i.e., adding entries to the zero sub-matrices in the discrete system. Residual-based stabilization techniques are widespread for flow problems~\cite{donea2003finite} and add a stabilization term based on the current residuum's value.
\par
We will use the continuous interior penalty (CIP) method, as proposed in~\cite{westerkamp2019finite} for the R13 system. This technique adds stabilization terms based on edge inner products (for two dimensions). The modified bilinear form \(\tilde{\mathcal{A}}\) then reads
\begin{align}
  \tilde{\mathcal{A}}\left(\te{\mathcal{U}_h},\te{\mathcal{V}_h}\right)
  &=
  \mathcal{A}\left(\te{\mathcal{U}_h},\te{\mathcal{V}_h}\right)
  +
  j_\theta(\theta_h,\kappa_h)
  +
  j_{\te{u}}(\te{u}_h,\te{v}_h)
  +
  j_p(p_h,q_h)
  ,
  \label{eq_weakFormStabilized}
\end{align}
where the stabilization terms are given by
\begin{align}
  j_\theta(\theta_h,\kappa_h)
  &=
  \delta_\theta \sum_\mathcal{E} \int_\mathcal{E} h^3 [\te{\nabla}\theta_h \te{\cdot} \te{n}] [\te{\nabla}\kappa_h \te{\cdot} \te{n}] \dd l
  ,
  \\
  j_{\te{u}}(\te{u}_h,\te{v}_h)
  &=
  \delta_{\te{u}} \sum_\mathcal{E} \int_\mathcal{E} h^3 [\te{\nabla}\te{u}_h \te{\cdot} \te{n}] \te{\cdot} [\te{\nabla}\te{v}_h \te{\cdot} \te{n}] \dd l
  ,
  \\
  j_p(p_h,q_h)
  &=
  \delta_p \sum_\mathcal{E} \int_\mathcal{E} h [\te{\nabla}p_h \te{\cdot} \te{n}] [\te{\nabla}q_h \te{\cdot} \te{n}] \dd l
  .
\end{align}
Here, \(\mathcal{E}\) is the index set of all interior element faces (i.e., edges) with \(\mathcal{E} \cap \partial \Omega = \emptyset\) and
\begin{equation}
  [\te{f} \cdot \te{n}] = \te{f}^+ \cdot \te{n}^+ + \te{f}^- \cdot \te{n}^-
  ,
\end{equation}
denotes the \(\te{f}\)-jump across the element boundary, weighted with the oppositely directed edge normals \(\te{n}^+\) and \(\te{n}^-\). Assuming the fields to be in \(C^1(\Omega)\) leads to no addition of stabilization terms. The method is, therefore, consistent~\cite{westerkamp2017curvature}. The different mesh size scalings result from an analysis in~\cite{westerkamp2017continous}, such that the order of stabilization does not change due to mesh refinement. The remaining parameters \(\delta_\theta,\delta_{\te{u}},\delta_p\) are not very sensitive, and the method is very robust to produce low errors for a wide range of \(\delta_\star\)-values.
Compare, for example, with~\cite{burman2006edgeStabilization}, where \citeauthor{burman2006edgeStabilization} presented a discussion of the CIP method applied to the generalized Stokes problem.
\par
With the presented stabilization, we have the following property of the system:
\begin{theorem}\label{thm_positiveDefiniteness}
  Consider a set of admissible system, stabilization, and boundary conditions as
  \begin{itemize}
    \item \(\Knud > 0\) to avoid division by zero,
    \item \(\delta_\theta,\delta_{\te{u}},\delta_p > 0 \) to avoid zero diagonals using stabilization,
    \item \(\tilde{\chi} > 0\) to have positive boundary terms in the diagonal sub-functionals,
    \item \(\epsilon^{\text{w}} \ge 0\) to guarantee non-negativity of all inflow boundary terms.
  \end{itemize}
  Then, the discrete stabilized weak form \(\tilde{\mathcal{A}}: \mathbb{H}_h \times \mathbb{H}_h \rightarrow \mathbb{R}, \left(\te{\mathcal{U}_h},\te{\mathcal{V}_h}\right) \mapsto \tilde{\mathcal{A}}\left(\te{\mathcal{U}_h},\te{\mathcal{V}_h}\right)\) is positive-definite for non-constant discrete fields \(\theta_h, \te{u}_h, p_h\).
\end{theorem}
\begin{proof}
  We use the notation \cref{eq_compundWeakFormAlternative} for \(\tilde{\mathcal{A}}\left(\te{\mathcal{U}_h},\te{\mathcal{V}_h}\right)\) and use antisymmetry of all non-diagonal sub-functionals to obtain
  \begin{equation}
    \tilde{\mathcal{A}}\left(\te{\mathcal{U}_h},\te{\mathcal{U}_h}\right)
    =
    a(\te{s}_h,\te{s}_h)
    +
    \bar{d}((\tee{\sigma}_h, p_h), (\tee{\sigma}_h, p_h))
    +
    j_\theta(\theta_h,\theta_h)
    +
    j_{\te{u}}(\te{u}_h,\te{u}_h)
    +
    j_p(p_h,p_h)
  \end{equation}
  For the stabilization terms, it holds by construction that
  \begin{equation}
    j_\theta(\theta_h,\theta_h)
    >0
    ,
    \quad
    j_{\te{u}}(\te{u}_h,\te{u}_h)
    >0
    ,
    \quad
    j_p(p_h,p_h)
    > 0
    ,
  \end{equation}
  for non-constant discrete fields and \(\delta_\star > 0\). The quadratic nature of the diagonal terms in \cref{eq_subf_a,eq_subf_dTilde} ensures positivity with
  \begin{equation}
    a(\te{s}_h,\te{s}_h) > 0 \; \forall \te{s}_h \ne \te{0}
    ,
    \quad
    \bar{d}((\tee{\sigma}_h, p_h), (\tee{\sigma}_h, p_h)) > 0 \; \forall (\tee{\sigma}_h, p_h) \ne (\te{0},0)
    .
  \end{equation}
  We can then directly follow that \(\tilde{\mathcal{A}}\left(\te{\mathcal{U}_h},\te{\mathcal{U}_h}\right) > 0 \; \forall \; \te{\mathcal{U}_h} \neq \te{0}\).
\end{proof}
\begin{remark}
  An analysis (similar to, e.g.,~\cite{burman2010interior}) would utilize the triple norm
  \begin{align}
    {||| \te{\mathcal{U}} |||}^2
    &=
    \frac{24}{25} \Knud \norm*{\mathrm{sym}(\te{\nabla}\te{s})}_{L^2,\Omega}^2
    + \frac{12}{25} \Knud \norm*{\mathrm{div}(\te{s})}_{L^2,\Omega}^2
    + \frac{4}{15} \frac{1}{\Knud} \norm*{\te{s}}_{L^2,\Omega}^2
    + \frac{1}{2} \frac{1}{\tilde{\chi}} \norm*{s_{n}}_{L^2,\Gamma}^2
    \nonumber
    \\
    &
    + \frac{12}{25} \tilde{\chi} \norm*{s_{t}}_{L^2,\Gamma}^2
    + \Knud \norm*{\text{stf}(\te{\nabla}\tee{\sigma})}_{L^2,\Omega}^2
    + \frac{4}{15} \frac{1}{\Knud} \norm*{\tee{\sigma}}_{L^2,\Omega}^2
    + \frac{9}{8} \tilde{\chi} \norm*{\sigma_{nn}}_{L^2,\Gamma}^2
    \nonumber
    \\
    &
    + \tilde{\chi} \norm*{ \sigma_{tt} + \frac{1}{2} \sigma_{nn} }_{L^2,\Gamma}^2
    + \frac{1}{\tilde{\chi}} \norm*{\sigma_{nt}}_{L^2,\Gamma}^2
    + \epsilon^{\text{w}} \tilde{\chi} \norm*{p + \sigma_{nn}}_{L^2,\Gamma}^2
    \nonumber
    \\
    &
    + j_\theta(\theta,\theta)
    + j_{\te{u}}(\te{u},\te{u})
    + j_p(p,p)
    ,
    \label{eq_tripleNorm}
  \end{align}
  where the \(L^2\)-scalar product \({(\star,\star)}_{D}\) over a domain \(D \subseteq \Omega\) defines the associated norm \({\lVert \star \rVert}_{L^2,D} = \sqrt{{(\star,\star)}_{D}}\). The stabilized weak form \cref{eq_weakFormStabilized} is coercive, using the norm \cref{eq_tripleNorm} with
  \begin{equation}
    \tilde{\mathcal{A}}\left(\te{\mathcal{U}_h},\te{\mathcal{U}_h}\right) \ge 1 \cdot {||| \te{\mathcal{U}}_h |||}^2 \quad \forall ~ \te{\mathcal{U}_h} \in \mathbb{H}_h
    .
  \end{equation}
\end{remark}

\section{Implementation}\label{s_implementationAndValidation}
The implementation of the compound weak form \cref{eq_compundWeakForm} uses the structured formulation of \cref{eq_discreteSystem} and schematically reads:
\begin{lstlisting}[
  style=pythonstyle,
  caption={Implementation of the stabilized compound weak form \(\tilde{\mathcal{A}}(\te{\mathcal{U}}_h,\te{\mathcal{V}}_h)\).},
  commentstyle=\color{comment_c}\rmfamily\itshape,
]
# 1) Left-hand sides, bilinear form A:
A[0] = a(s, r)     - b(theta, r) - c(r, sigma)   + 0         + 0
A[1] = b(kappa, s) + 0           + 0             + 0         + 0
A[2] = c(s, psi)   + 0           + d(sigma, psi) - e(u, psi) + f(p, psi)
A[3] = 0           + 0           + e(v, sigma)   + 0         + g(p, v)
A[4] = 0           + 0           + f(q, sigma)   - g(q, u)   + h(p, q)
# 2) Right-hand sides, linear functional L:
L[0] = - sum([
  n(r) * bcs[bc]["theta_w"] * df.ds(bc)
  for bc in bcs.keys()
])
# [...]
self.form_lhs = sum(A) + cip * j_theta(theta, kappa) + j_u(u, v) + j_p(p, q)
self.form_rhs = sum(L)
# [...]
df.solve(self.form_lhs == self.form_rhs, sol, [])
\end{lstlisting}
Here, we see the one-to-one correspondence between the underlying mathematics and the resulting source code. There is also no need to supply Dirichlet boundary conditions (observe ``\texttt{[]}'') to the ``\texttt{df.solve}''-routine. The weak formulation includes all boundary equations naturally in a weak sense. However, the sub-functionals still contain higher-order differential operators, e.g., the symmetric and trace-free part of a 2-tensor Jacobian \({(\te{\nabla}\tee{\sigma})}_{\text{stf}}\). For such higher-order tensors, not all required operators are available in the UF language 2019.1.0 (e.g., ``\texttt{Deviatoric}'' and ``\texttt{Trace}'' in ``\texttt{tensoralgebra.py}'' of~\cite{fenics2020uflRepo}). The same also applies to the DOLFIN C++ interface that uses the same underlying UFL to compile weak forms with the FFC\@.
\par
Therefore, we use the extension capabilities of the UFL and define the required operators using Einstein's index notation. This section first presents the important implementation aspects, focusing mainly on additional differential operators in \cref{ss_tensorFEMinFenics} and CIP stabilization in \cref{ss_cipStabilizationInFenics}.

\subsection{Tensorial Mixed FEM in FEniCS}\label{ss_tensorFEMinFenics}
An important implementation detail is treating the symmetric and trace-free operator for a tensor rank of two. For strictly two-dimensional problems, it is possible to use the default built-in ``\texttt{sym}''-, ``\texttt{tr}''- and ``\texttt{dev}''-operators of FEniCS/UFL successively. However, assuming a three-dimensional and \(z\)-homogenous problem (as stated in~\cref{ss_linearization}) requires a change in the operator. Following~\cite{fenics2019uflDocumentation}, UFL defines the deviatoric part of a 2-tensor as
\begin{equation}
  {(\tee{A})}_{\mathrm{dev}} = \tee{A} - \frac{A_{ii}}{d} \tee{I},
\end{equation}
using the dimension \(d\) and Einstein's summation notation for \(A_{ii}\). Performing computations on a two-dimensional mesh \(\Omega \in \mathbb{R}^2\) leads to \(d=2\). However, in our work, we assume homogeneity in the third spatial dimension (such that none of the relevant fields depends on the \(z\)-coordinate). A modified STF-operator for 2-tensors is, thus, used to obtain the definition of \cref{eq_rank2stf}. For \(\tee{A} \in \mathbb{R}^{2 \times 2}\), the modified operator is, for our purposes, defined as
\begin{equation}\label{eq:heatStfModification}
  {(\tee{A})}_{\text{stf}} = \frac{1}{2} \left( \tee{A} + \tee{A}^{\mathrm{T}} \right) - \frac{A_{ii}}{3} \tee{I}.
\end{equation}
The corresponding implementation of \({(\star)}_{\text{stf}}\), therefore, artificially assumes \(d=3\) and reads:
\begin{lstlisting}[
  style=pythonstyle,
  caption={3D STF-operator for 2D 2-tensors.},
  commentstyle=\color{comment_c}\rmfamily\itshape,
]
def stf3d2(rank2_2d):
    symm = 1/2 * (rank2_2d + ufl.transpose(rank2_2d))
    return symm - (1/3) * ufl.tr(symm) * ufl.Identity(2)
\end{lstlisting}
\par
The construction of the Frobenius inner product of two 3-tensors (as in \(\text{stf}(\te{\nabla}\tee{\sigma}) \teee{\because} \text{stf}(\te{\nabla}\tee{\psi})\)) needs more auxiliary functions. It is possible to implement the system component-wise and solve for \(p,u_x,u_y,\sigma_{xx},\sigma_{xy},\sigma_{yy}\) after expanding all operators in the weak form \cref{eq_subf_d}. However, this would increase the complexity even more and is error-prone. Using a computer algebra system to calculate \({(\nabla \tee{\sigma})}_{\text{stf}}\) reveals 18 different terms for only one tensorial expression. Therefore, we will again make extensive use of the tensor capabilities provided by FEniCS and UFL to avoid computing such expressions and have the corresponding source code in a compact form. In fact, up to second-order tensors, all UFL operators are intuitive, except for the already discussed 3D information encoding. However, dealing with 3-tensors is not straight-forward because FEniCS lacks implementations of the required operators in its current version.
\par
First of all, using a two-dimensional mesh leads to creating only two spatial variables acting in the differential operators. To respect the highest-order moments' shape assumptions \cref{eq:stressTensorShapes}, we define a lifting operator \(L\), mapping a 2D 2-tensor artificially to a trace-free 3D 2-tensor. The definition of
\begin{equation}
  L : \mathbb{R}^{2 \times 2} \rightarrow \mathbb{R}_{\mathrm{TF}}^{3 \times 3}
  ,
  \begin{pmatrix}
    a & b \\
    c & d
  \end{pmatrix}
  \mapsto
  \begin{pmatrix}
    a & b & 0 \\
    c & d & 0 \\
    0 & 0 & -(a+d)
  \end{pmatrix}
  ,
\end{equation}
implements as:
\begin{lstlisting}[
  style=pythonstyle,
  caption={Custom operator to lift a 2D 2-tensor to a 3D STF 2-tensor.},
  commentstyle=\color{comment_c}\rmfamily\itshape,
]
def gen3dTF2(rank2_2d):
    return df.as_tensor([
        [rank2_2d[0, 0], rank2_2d[0, 1], 0],
        [rank2_2d[1, 0], rank2_2d[1, 1], 0],
        [0, 0, -rank2_2d[0, 0]-rank2_2d[1, 1]]
    ])
\end{lstlisting}
The gradient operator is extended in a similar fashion accounting for the third dimension as:
\begin{lstlisting}[
  style=pythonstyle,
  caption={Custom gradient operator to account for three dimensions.},
  commentstyle=\color{comment_c}\rmfamily\itshape,
]
def grad3dOf2(rank2_3d):
    grad2d = df.grad(rank2_3d)
    dim3 = df.as_tensor([[0, 0, 0], [0, 0, 0], [0, 0, 0]])
    grad3d = df.as_tensor([grad2d[:, :, 0], grad2d[:, :, 1], dim3[:, :]])
    return grad3d
\end{lstlisting}
With these operators at hand, it is now possible to evaluate \({(\nabla \tee{\sigma})}_{\text{stf}}\). We use the definition \cref{eq_rank3stf} of the symmetric and trace-free part of a 3-tensor directly in FEniCS, including all Einstein summation conventions. The implementation of the corresponding function then reads:
\begin{lstlisting}[
  style=pythonstyle,
  caption={Custom operator to obtain the STF-part of a 3-tensor.},
  commentstyle=\color{comment_c}\rmfamily\itshape,
]
def stf3d3(rank3_3d):
    i, j, k, l = ufl.indices(4)
    delta = df.Identity(3)
    sym_ijk = sym3d3(rank3_3d)[i, j, k]
    traces_ijk = 1/5 * (
        + sym3d3(rank3_3d)[i, l, l] * delta[j, k]
        + sym3d3(rank3_3d)[l, j, l] * delta[i, k]
        + sym3d3(rank3_3d)[l, l, k] * delta[i, j]
    )
    tracefree_ijk = sym_ijk - traces_ijk
    return ufl.as_tensor(tracefree_ijk, (i, j, k))
\end{lstlisting}
Here, the symmetric part of a 3-tensor is the average of all possible transpositions, as defined in \cref{eq_rank3sym}, and is translated into UFL code using:
\begin{lstlisting}[
  style=pythonstyle,
  caption={Custom operator to obtain the symmetric part of a 3-tensor.},
  commentstyle=\color{comment_c}\rmfamily\itshape,
]
def sym3d3(rank3_3d):
    i, j, k = ufl.indices(3)
    symm_ijk = 1/6 * (
        # All permutations
        + rank3_3d[i, j, k] + rank3_3d[i, k, j] + rank3_3d[j, i, k]
        + rank3_3d[j, k, i] + rank3_3d[k, i, j] + rank3_3d[k, j, i]
    )
    return ufl.as_tensor(symm_ijk, (i, j, k))
\end{lstlisting}
These auxiliary functions are implemented in a separate ``\texttt{tensoroperations}''-module and allow the desired one-to-one correlation between the mathematical formulation and the corresponding implementation of the weak formulation. In general, the summation convention capabilities of FEniCS would also allow tackling even higher-order moment equations, such as the R26 equations~\cite{gu2009high}, if auxiliary \(n\)-tensor operators are defined. The implementation of \(d(\sigma,\psi)\), as the most complex bilinear form, is then schematically obtained as:
\begin{lstlisting}[
  style=pythonstyle,
  caption={Schemtic implementation of the bilinear form \(d(\sigma,\psi)\).},
  commentstyle=\color{comment_c}\rmfamily\itshape,
]
def d(sigma_, psi_):
return (
    kn * df.inner(
        to.stf3d3(to.grad3dOf2(to.gen3dTF2(sigma_))),
        to.stf3d3(to.grad3dOf2(to.gen3dTF2(psi_)))
    )
    + (1/(2*kn)) * df.inner(
        to.gen3dTF2(sigma_), to.gen3dTF2(psi_)
    )
) * df.dx # + [...]
\end{lstlisting}

\subsection{CIP Stabilization in FEniCS}\label{ss_cipStabilizationInFenics}
The CIP stabilization method's implementation uses the support for discontinuous Galerkin (DG) operators in the UFL\@. Integration of all interior edges, as a subset of all edges, can be archived using ``\texttt{dS}'' instead of ``\texttt{ds}'' (which only acts on boundary edges). The resulting implementation then intuitively reads, e.g., for the scalar and vector stabilization functionals \(j_\theta\), \(j_{\te{u}}\):
\begin{lstlisting}[
  style=pythonstyle,
  caption={Implementation of CIP stabilization.},
  commentstyle=\color{comment_c}\rmfamily\itshape,
]
# Define custom measeasures for boundary edges and inner edges
df.ds = df.Measure("ds", domain=mesh, subdomain_data=boundaries)
df.dS = df.Measure("dS", domain=mesh, subdomain_data=boundaries)
# Define mesh measuers
h_msh = df.CellDiameter(mesh)
h_avg = (h_msh("+") + h_msh("-"))/2.0
# 3) CIP Stabilization:
def j_theta(theta, kappa):
    return (
        + delta_theta * h_avg**3 *
        df.jump(df.grad(theta), n_vec) * df.jump(df.grad(kappa), n_vec)
    ) * df.dS
def j_u(u, v):
    return (
        + delta_u * h_avg**3 *
        df.dot(df.jump(df.grad(u), n_vec), df.jump(df.grad(v), n_vec))
    ) * df.dS
\end{lstlisting}

\section{Convergence Study Based on Mesh Refinement}\label{s_validation}
We perform a convergence study to validate the numerical method, comparing the discrete solutions to their exact solutions. The solver repository~\cite{theisen2020fenicsr13Zenodo} includes all exact solutions for reproducibility. Computations on a series of refined meshes reveal the method's numerical convergence properties and show convergence with increased mesh resolution.

\subsection{Computational Domain and Test Case Parameters}\label{ss_convergenceStudy}
We consider the ring domain \(\Omega \subset \mathbb{R}^2\) as the area between two coaxial circles with radii \(R_1\) and \(R_2\) as
\begin{equation}
  \Omega = \left\{ \te{x} = {(x,y)}^{\mathrm{T}} : R_1 \le {\lVert \te{x} \rVert}_2 \le R_2 \right\}
  .
\end{equation}
We choose \(R_1=0.5\) and \(R_2=2\), which follows previous works~\cite{torrilhon2017hierarchical,westerkamp2019finite}. The inner boundary corresponds to \(\Gamma_1\) and the outer circle to \(\Gamma_2\). The particular domain \(\Omega\) avoids sharp corners, and a prescription of boundary fluxes does not produce any problems because the circle origin \((0,0)\) is not part of the domain. The computation of exact solutions is, therefore, possible. As already mentioned, we assume 2D problems as simplifications for 3D problems with full symmetry and homogeneity in the third spatial direction.
\par
To test the numerical method, we consider a series of general and unstructured triangular meshes without spatial refinement. As a result, the discretized domain contains approximately similar cell sizes. This very general setup does not take any properties of curved domains into account. The mesh resolution at the inner boundary, for example, is equal to the resolution at the outer boundary, although the curvatures are not equal.
This general approach allows us to test the numerical method for principal correctness using the most general mesh types.
\par
We use the mesh generator Gmsh~\cite{geuzaine2009gmsh} to create a series of ring meshes. Note that no element split is applied to obtain the refined meshes, and finer meshes result from a complete re-meshing procedure with a lower cell size factor. The maximum cell size \(h_{\max}\) characterizes the mesh resolution. The \cref{fig:heatMeshes} presents exemplarily meshes with their corresponding \(h_{\max}\). In contrast to~\cite{westerkamp2019finite}, FEniCS cannot utilize isoparametric higher-order boundary representations in the current implementation. \(L^2\)-convergence rates beyond second-order are therefore not expected.
\par
We perform all calculations in a Docker container on an iMac 2017 with a 3.4 GHz Intel Core i5--7500 CPU and 48GiB memory. The resulting discrete systems were solved with the direct solver MUMPS~\cite{MUMPS:1,MUMPS:2}, shipped with FEniCS\@.
\begin{figure}[t]
  \begin{subfigure}[c]{0.24\textwidth}
    \includegraphics[width=\linewidth]{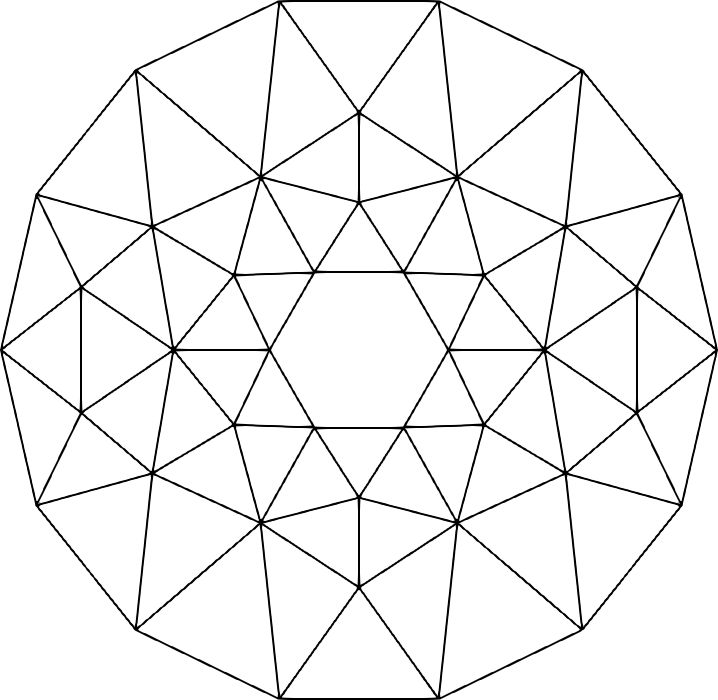}
    \subcaption{
      \(h_{\max}=0.99\).
    }\label{sfig:mesh0}
  \end{subfigure}
  \hfill
  \begin{subfigure}[c]{0.24\textwidth}
    \includegraphics[width=\linewidth]{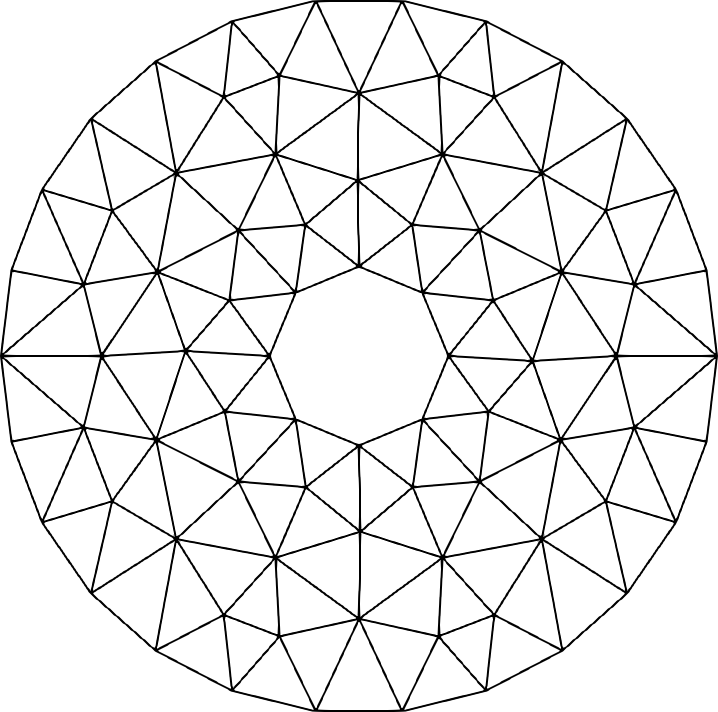}
    \subcaption{
      \(h_{\max}=0.63\).
    }\label{sfig:mesh1}
  \end{subfigure}
  \hfill
  \begin{subfigure}[c]{0.24\textwidth}
    \includegraphics[width=\linewidth]{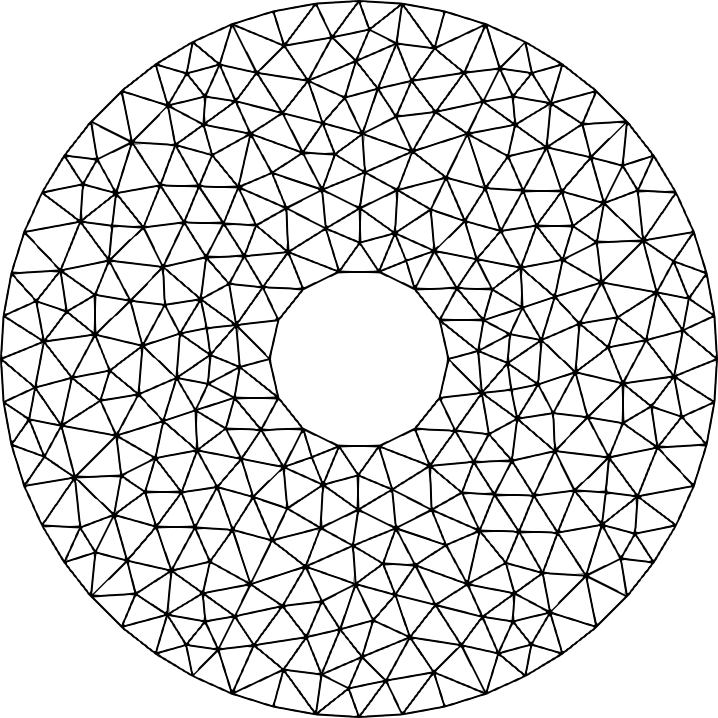}
    \subcaption{
      \(h_{\max}=0.33\).
    }\label{sfig:mesh2}
  \end{subfigure}
  \hfill
  \begin{subfigure}[c]{0.24\textwidth}
    \includegraphics[width=\linewidth]{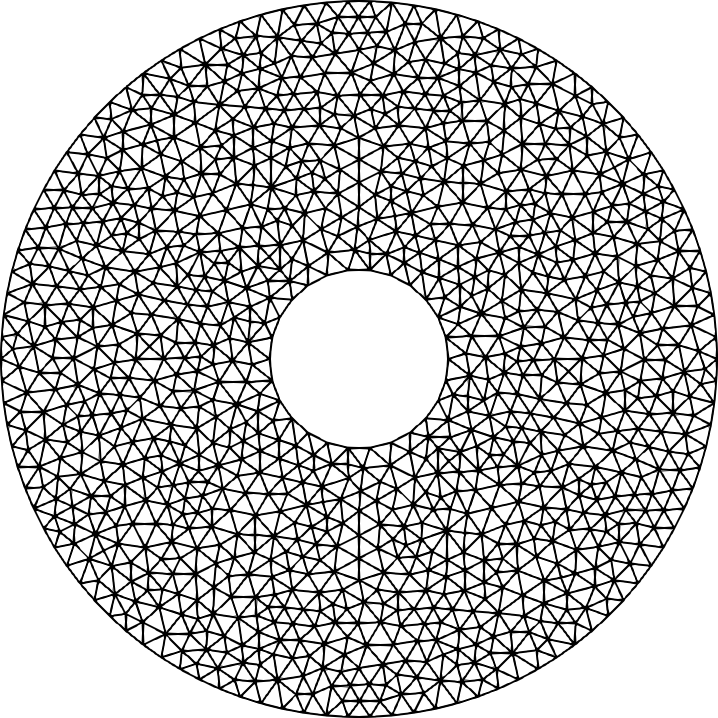}
    \subcaption{
      \(h_{\max}=0.17\).
    }\label{sfig:mesh3}
  \end{subfigure}
  \caption{
    Series of unstructured triangular meshes used for the convergence study: Note that the coarsest mesh is not uniform to sufficiently resolve the inner boundary. With finer meshes, this effect vanishes. The bounding box of the mesh slightly varies due to the curved boundary and the node placement.
  }\label{fig:heatMeshes}
  \Description{
    Series of unstructured triangular meshes used for the convergence study: Note that the coarsest mesh is not uniform to sufficiently resolve the inner boundary. With finer meshes, this effect vanishes. The bounding box of the mesh slightly varies due to the curved boundary and the node placement.
  }
\end{figure}
\paragraph{Error Measures}
To rate the success of the numerical method, we introduce relevant error measures. The standard relative \(L^2\)-function error is
\begin{equation}
  e_{L^2} = \frac{
    {\lVert f_{\mathrm{ex}}-f_h\rVert}_{L^2(\Omega)}
  }{\max{\left\{ f_{\mathrm{ex}}|_n\right\}}_{n \in \eta}},
\end{equation}
where \(f_{\mathrm{ex}}\) denotes the exact solution, and  \(\eta\) is the set of all mesh nodes.
\par
The typical way to obtain a discrete solution is in terms of node values. Therefore, particular interest should also be given to these particular points, using a relative error \(e_{l^\infty}\). We define this vector-based error as
\begin{equation}
  e_{l^\infty} = \frac{
    {\lVert {\left\{ f_{\mathrm{ex}}|_n-f_h|_n \right\}}_{n \in \eta} \rVert}_{l^\infty(\eta)}
    }{\max{\left\{ f_{\mathrm{ex}}|_n\right\}}_{n \in \eta}},
\end{equation}
while \(l^\infty\) is used to indicate that this error only considers point-wise errors based on mesh node values. If \(e_{l^\infty}\) decays to zero for refined meshes, we have ensured that for all points of interest \textendash\ i.e., the mesh nodes \textendash\ the solution converges towards the exact solution.

\subsection{Homogenous Flow Around Cylinder}
The test case considers a flow scenario with inflow and outflow boundary conditions, similar to~\cite{torrilhon2017hierarchical}. The velocity prescription coefficient, therefore, is \(\epsilon^{\mathrm{w}} \ne 0\). The outer wall is not impermeable but only acts as a cut-off from a larger homogenous velocity field. The inner cylinder wall is modeled as a non-rotating impermeable wall with zero velocity in normal direction \(u_n^{\mathrm{w}}|_{\Gamma_1}=0\) and zero tangential velocity \(u_t^{\mathrm{w}}|_{\Gamma_1}=0\) prescribed. A temperature difference is applied between the inner and the outer cylinder walls with \(\theta^\mathrm{w}|_{\Gamma_1} = 1\) and \(\theta^\mathrm{w}|_{\Gamma_2} = 2\) to render the case more complicated.
\par
A pressure difference drives the flow at the outer cylinder wall with \(p^{\mathrm{w}}|_{\Gamma_2} = - p_0 n_x\), in which the background pressure is set to \(p_0=0.27\), and \(n_x\) is equal to \(\cos(\phi)\) for the considered geometry. The velocity components at the outer boundary are set to \(u_n^{\mathrm{w}}|_{\Gamma_2} = u_0 n_x\) and \(u_t^{\mathrm{w}}|_{\Gamma_2} = - u_0 n_y\) with background velocity \(u_0 = 1\) and \(n_y = \sin(\phi)\). The remaining parameters read \(\epsilon^{\mathrm{w}}|_{\Gamma_1} = 10^{-3}\) to focus on velocity prescription, \(\epsilon^{\mathrm{w}}|_{\Gamma_2} = 10^{3}\) to focus on pressure prescription, \(\Knud = 1\) as flow characterization, and \(\tilde{\chi}=1\) as wall parameter.
\par
We will first consider \(\mathbb{P}_2\mathbb{P}_1\mathbb{P}_1\mathbb{P}_2\mathbb{P}_1\) elements (corresponding to the fields \(\tee{\sigma},\te{u},p,\te{s},\theta\) in order) without CIP stabilization. One could see this element combination as a generalization to classical Taylor--Hood elements with two hierarchies of moment sets: the heat-system and the stress-system variables. We choose \(k+1\) for both systems as the polynomial order for the highest-order moments (\(\tee{\sigma}\) and \(\te{s}\)) and \(k\) for all other fields. The resulting errors in \cref{fig:r13_1_coeffs_nosources_norot_inflow_p1p2p1p1p2_nostab} show an almost optimal convergence in the \(L^2\)-error measure. The node-values in the \(l^\infty\)-error also convergence but with a decreased rate for velocity and temperature fields.
\par
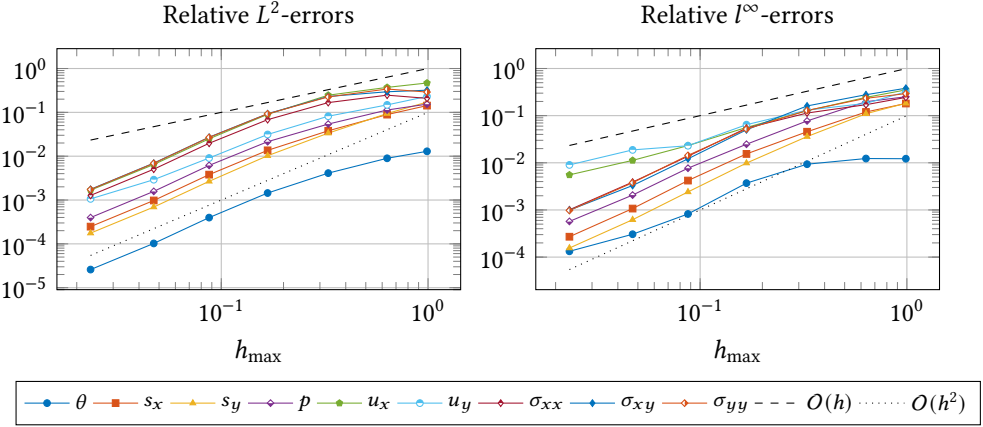
\begin{figure}[t]
  \newcommand{\datapath}{./data}%
  \newcommand{\firstOrderFactor}{1}%
  \newcommand{\secondOrderFactor}{0.1}%
  \newcommand{\errorfile}{convergence/article/r13_1_coeffs_nosources_norot_inflow_p1p2p1p1p2_nostab/errors.csv}%
  \newcommand{\fieldList}{theta/\(\theta\),sx/\(s_x\),sy/\(s_y\),p/\(p\),ux/\(u_x\),uy/\(u_y\),sigmaxx/\(\sigma_{xx}\),sigmaxy/\(\sigma_{xy}\),sigmayy/\(\sigma_{yy}\)}%
  \centering
  \begin{tikzpicture}[
]
  \providecommand{\datapath}{../../../data}
  \providecommand{\errortype}{L_2}
  \providecommand{\errorfile}{convergence/r13_1_coeffs_nosources_norot_inflow_p1p1p1p1p1_stab/errors.csv}
  \providecommand{\fieldList}{theta/\(\theta\),sx/\(s_x\),sy/\(s_y\),p/\(p\),ux/\(u_x\),uy/\(u_y\),sigmaxx/\(\sigma_{xx}\),sigmaxy/\(\sigma_{xy}\),sigmayy/\(\sigma_{yy}\)}
  \providecommand{\firstOrderFactor}{1}
  \providecommand{\secondOrderFactor}{0.2}

  \definecolor{color1}{rgb}{0,    ,0.4470,0.7410}
  \definecolor{color2}{rgb}{0.8500,0.3250,0.0980}
  \definecolor{color3}{rgb}{0.9290,0.6940,0.1250}
  \definecolor{color4}{rgb}{0.4940,0.1840,0.5560}
  \definecolor{color5}{rgb}{0.4660,0.6740,0.1880}
  \definecolor{color6}{rgb}{0.3010,0.7450,0.9330}
  \definecolor{color7}{rgb}{0.6350,0.0780,0.1840}
  \pgfplotscreateplotcyclelist{matlab}{
    color1,every mark/.append style={solid},mark=*\\
    color2,every mark/.append style={solid},mark=square*\\
    color3,every mark/.append style={solid},mark=triangle*\\
    color4,every mark/.append style={solid},mark=halfsquare*\\
    color5,every mark/.append style={solid},mark=pentagon*\\
    color6,every mark/.append style={solid},mark=halfcircle*\\
    color7,every mark/.append style={solid,rotate=180},mark=halfdiamond*\\
    color1,every mark/.append style={solid},mark=diamond*\\
    color2,every mark/.append style={solid},mark=halfsquare right*\\
    color3,every mark/.append style={solid},mark=halfsquare left*\\
  }

  \tikzset{
    mark options={
      mark size=1.25,
    }
  }

  \pgfplotstableread[col sep=comma, sci, precision=10]{\datapath/\errorfile}{\datatable}
  \pgfplotstablegetrowsof{\datatable} 
  \pgfmathsetmacro{\numRows}{\pgfplotsretval-1}
  \pgfmathsetmacro{\numRowsM}{\numRows-1}
  \pgfmathsetmacro{\errMax}{0}
  \pgfmathsetmacro{\errMin}{10000}

  \pgfplotstablegetelem{0}{h}\of{\datatable}
  \pgfmathsetmacro{\hFirst}{\pgfplotsretval}
  \pgfplotstablegetelem{\numRows}{h}\of{\datatable}
  \pgfmathsetmacro{\hLast}{\pgfplotsretval}
  \pgfplotstablegetelem{\numRowsM}{h}\of{\datatable}
  \pgfmathsetmacro{\hLastM}{\pgfplotsretval}

  \begin{groupplot}[
    group style={
      group size= 2 by 1,
    },
    xlabel={\(h_{\max}\)},
    grid=major,
    cycle list name=matlab,
    width=\textwidth/2,
    xmode=log,
    ymode=log,
    ytick distance={10},
    height=4.75cm,
  ]
    \nextgroupplot[
      title={Relative \(L^2\)-errors},
      legend to name={CommonLegend_\errorfile},
      legend style={
        legend columns=-1,
        font=\footnotesize
      }
    ]
      \foreach \i/\j in \fieldList {
        \addplot table[x=h, y=\i_L_2, col sep=comma]{\datatable};
        \addlegendentryexpanded{\j} 
      }
      \addplot [dashed, domain=\hLast:\hFirst] {\firstOrderFactor*x^1};
      \addlegendentryexpanded{$\mathcal{O}(h)$}
      \addplot [dotted, domain=\hLast:\hFirst] {\secondOrderFactor*x^2};
      \addlegendentryexpanded{$\mathcal{O}(h^2)$}
    \nextgroupplot[
      title={Relative \(l^\infty\)-errors}
    ]
      \foreach \i/\j in \fieldList {
        \addplot table[x=h, y=\i_l_inf, col sep=comma]{\datatable};
      }
      \addplot [dashed, domain=\hLast:\hFirst] {\firstOrderFactor*x^1};
      \addplot [dotted, domain=\hLast:\hFirst] {\secondOrderFactor*x^2};

    \end{groupplot}
    \coordinate (c3) at ($(group c1r1)!.5!(group c2r1)$); \node[below] at (c3 |- current bounding box.south) {\pgfplotslegendfromname{CommonLegend_\errorfile}};
\end{tikzpicture}
  \caption{
      Relative errors using unstabilized \(\mathbb{P}_2\mathbb{P}_1\mathbb{P}_1\mathbb{P}_2\mathbb{P}_1\) elements for the homogenous flow around a cylinder: Almost all fields have second-order convergence rates in the \(L^2\)-norm. In the \(l^\infty\)-norm, reduced rates are observed, but at least first-order convergence is guaranteed.
    }\label{fig:r13_1_coeffs_nosources_norot_inflow_p1p2p1p1p2_nostab}
    \Description{
      Relative errors using unstabilized \(\mathbb{P}_2\mathbb{P}_1\mathbb{P}_1\mathbb{P}_2\mathbb{P}_1\) elements for the homogenous flow around a cylinder: Almost all fields have second-order convergence rates in the \(L^2\)-norm. In the \(l^\infty\)-norm, reduced rates are observed, but at least first-order convergence is guaranteed.
    }
\end{figure}
\par
However, an increased discretization order for the highest-order fields is often not desired from an engineering perspective. For example, knowing the flow field's velocity gives more practical insight than knowing its stress tensor. We, therefore, aim also to use equal-order \(\mathbb{P}_1\) elements using the proposed CIP stabilization. The resulting errors for the same test case, using the stabilized setup, are presented in \cref{fig:r13_1_coeffs_nosources_norot_inflow_p1p1p1p1p1_stab}. The set of stabilization parameters \(\delta_\theta=1,\delta_{\te{u}}=1, \delta_p=0.01\) stems from~\cite{westerkamp2019finite}. We observe a decrease in relative accuracy and convergence order for the \(\theta\)-field Compared to the unstabilized setup. Parameter tuning for the \(\delta_\star\)-values might improve the numerical properties even more.
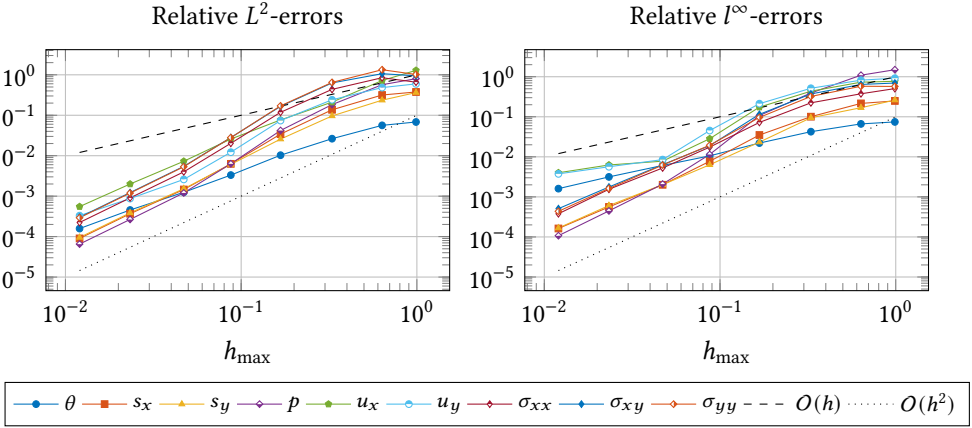
\begin{figure}[t]
  \newcommand{\datapath}{./data}%
  \newcommand{\firstOrderFactor}{1}%
  \newcommand{\secondOrderFactor}{0.1}%
  \newcommand{\errorfile}{convergence/article/r13_1_coeffs_nosources_norot_inflow_p1p1p1p1p1_stab/errors.csv}%
  \newcommand{\fieldList}{theta/\(\theta\),sx/\(s_x\),sy/\(s_y\),p/\(p\),ux/\(u_x\),uy/\(u_y\),sigmaxx/\(\sigma_{xx}\),sigmaxy/\(\sigma_{xy}\),sigmayy/\(\sigma_{yy}\)}%
  \centering
  \begin{tikzpicture}[
]
  \providecommand{\datapath}{../../../data}
  \providecommand{\errortype}{L_2}
  \providecommand{\errorfile}{convergence/r13_1_coeffs_nosources_norot_inflow_p1p1p1p1p1_stab/errors.csv}
  \providecommand{\fieldList}{theta/\(\theta\),sx/\(s_x\),sy/\(s_y\),p/\(p\),ux/\(u_x\),uy/\(u_y\),sigmaxx/\(\sigma_{xx}\),sigmaxy/\(\sigma_{xy}\),sigmayy/\(\sigma_{yy}\)}
  \providecommand{\firstOrderFactor}{1}
  \providecommand{\secondOrderFactor}{0.2}

  \definecolor{color1}{rgb}{0,    ,0.4470,0.7410}
  \definecolor{color2}{rgb}{0.8500,0.3250,0.0980}
  \definecolor{color3}{rgb}{0.9290,0.6940,0.1250}
  \definecolor{color4}{rgb}{0.4940,0.1840,0.5560}
  \definecolor{color5}{rgb}{0.4660,0.6740,0.1880}
  \definecolor{color6}{rgb}{0.3010,0.7450,0.9330}
  \definecolor{color7}{rgb}{0.6350,0.0780,0.1840}
  \pgfplotscreateplotcyclelist{matlab}{
    color1,every mark/.append style={solid},mark=*\\
    color2,every mark/.append style={solid},mark=square*\\
    color3,every mark/.append style={solid},mark=triangle*\\
    color4,every mark/.append style={solid},mark=halfsquare*\\
    color5,every mark/.append style={solid},mark=pentagon*\\
    color6,every mark/.append style={solid},mark=halfcircle*\\
    color7,every mark/.append style={solid,rotate=180},mark=halfdiamond*\\
    color1,every mark/.append style={solid},mark=diamond*\\
    color2,every mark/.append style={solid},mark=halfsquare right*\\
    color3,every mark/.append style={solid},mark=halfsquare left*\\
  }

  \tikzset{
    mark options={
      mark size=1.25,
    }
  }

  \pgfplotstableread[col sep=comma, sci, precision=10]{\datapath/\errorfile}{\datatable}
  \pgfplotstablegetrowsof{\datatable} 
  \pgfmathsetmacro{\numRows}{\pgfplotsretval-1}
  \pgfmathsetmacro{\numRowsM}{\numRows-1}
  \pgfmathsetmacro{\errMax}{0}
  \pgfmathsetmacro{\errMin}{10000}

  \pgfplotstablegetelem{0}{h}\of{\datatable}
  \pgfmathsetmacro{\hFirst}{\pgfplotsretval}
  \pgfplotstablegetelem{\numRows}{h}\of{\datatable}
  \pgfmathsetmacro{\hLast}{\pgfplotsretval}
  \pgfplotstablegetelem{\numRowsM}{h}\of{\datatable}
  \pgfmathsetmacro{\hLastM}{\pgfplotsretval}

  \begin{groupplot}[
    group style={
      group size= 2 by 1,
    },
    xlabel={\(h_{\max}\)},
    grid=major,
    cycle list name=matlab,
    width=\textwidth/2,
    xmode=log,
    ymode=log,
    ytick distance={10},
    height=4.75cm,
  ]
    \nextgroupplot[
      title={Relative \(L^2\)-errors},
      legend to name={CommonLegend_\errorfile},
      legend style={
        legend columns=-1,
        font=\footnotesize
      }
    ]
      \foreach \i/\j in \fieldList {
        \addplot table[x=h, y=\i_L_2, col sep=comma]{\datatable};
        \addlegendentryexpanded{\j} 
      }
      \addplot [dashed, domain=\hLast:\hFirst] {\firstOrderFactor*x^1};
      \addlegendentryexpanded{$\mathcal{O}(h)$}
      \addplot [dotted, domain=\hLast:\hFirst] {\secondOrderFactor*x^2};
      \addlegendentryexpanded{$\mathcal{O}(h^2)$}
    \nextgroupplot[
      title={Relative \(l^\infty\)-errors}
    ]
      \foreach \i/\j in \fieldList {
        \addplot table[x=h, y=\i_l_inf, col sep=comma]{\datatable};
      }
      \addplot [dashed, domain=\hLast:\hFirst] {\firstOrderFactor*x^1};
      \addplot [dotted, domain=\hLast:\hFirst] {\secondOrderFactor*x^2};

    \end{groupplot}
    \coordinate (c3) at ($(group c1r1)!.5!(group c2r1)$); \node[below] at (c3 |- current bounding box.south) {\pgfplotslegendfromname{CommonLegend_\errorfile}};
\end{tikzpicture}
  \caption{
      Relative errors using stabilized \(\mathbb{P}_1\) equal-order elements for the homogenous flow around a cylinder: All fields except for \(\theta\) have second-order convergence rates in the \(L^2\)-norm. In the \(l^\infty\)-norm, reduced rates are observed for the stabilized fields \(\theta,\te{u}\).
    }\label{fig:r13_1_coeffs_nosources_norot_inflow_p1p1p1p1p1_stab}
    \Description{
      Relative errors using stabilized \(\mathbb{P}_1\) equal-order elements for the homogenous flow around a cylinder: All fields except for \(\theta\) have second-order convergence rates in the \(L^2\)-norm. In the \(l^\infty\)-norm, reduced rates are observed for the stabilized fields \(\theta,\te{u}\).
    }
\end{figure}
\par
The convergence behavior is improved using \(\mathbb{P}_2\) equal-order elements in \cref{fig:r13_1_coeffs_nosources_norot_inflow_p2p2p2p2p2_stab}. We now have second-order convergence for all fields in \(L^2\). However, using only a first-order boundary approximation for a curved domain limits the convergence rate, as discussed in~\cite{westerkamp2019finite}. An inspection of the discrete solution reveals the dominant error at the inner curved boundary, confirming the above considerations.
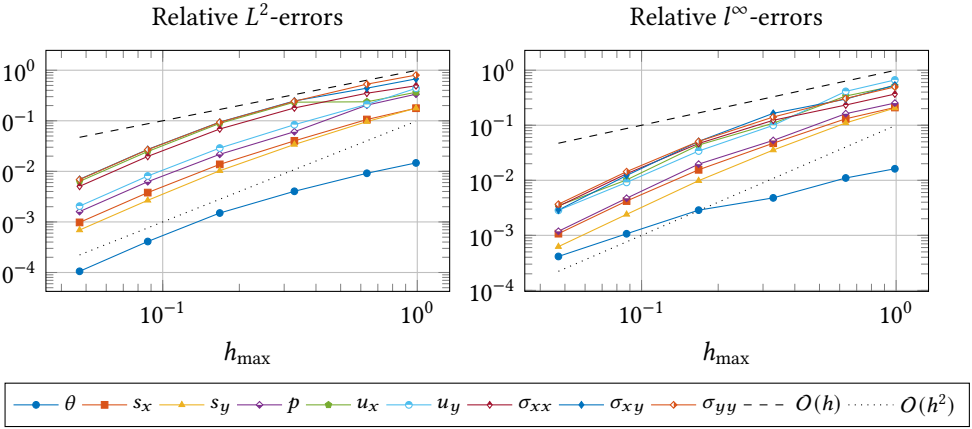
\begin{figure}[t]
  \newcommand{\datapath}{./data}%
  \newcommand{\firstOrderFactor}{1}%
  \newcommand{\secondOrderFactor}{0.1}%
  \newcommand{\errorfile}{convergence/article/r13_1_coeffs_nosources_norot_inflow_p2p2p2p2p2_stab/errors.csv}%
  \newcommand{\fieldList}{theta/\(\theta\),sx/\(s_x\),sy/\(s_y\),p/\(p\),ux/\(u_x\),uy/\(u_y\),sigmaxx/\(\sigma_{xx}\),sigmaxy/\(\sigma_{xy}\),sigmayy/\(\sigma_{yy}\)}%
  \centering
  \begin{tikzpicture}[
]
  \providecommand{\datapath}{../../../data}
  \providecommand{\errortype}{L_2}
  \providecommand{\errorfile}{convergence/r13_1_coeffs_nosources_norot_inflow_p1p1p1p1p1_stab/errors.csv}
  \providecommand{\fieldList}{theta/\(\theta\),sx/\(s_x\),sy/\(s_y\),p/\(p\),ux/\(u_x\),uy/\(u_y\),sigmaxx/\(\sigma_{xx}\),sigmaxy/\(\sigma_{xy}\),sigmayy/\(\sigma_{yy}\)}
  \providecommand{\firstOrderFactor}{1}
  \providecommand{\secondOrderFactor}{0.2}

  \definecolor{color1}{rgb}{0,    ,0.4470,0.7410}
  \definecolor{color2}{rgb}{0.8500,0.3250,0.0980}
  \definecolor{color3}{rgb}{0.9290,0.6940,0.1250}
  \definecolor{color4}{rgb}{0.4940,0.1840,0.5560}
  \definecolor{color5}{rgb}{0.4660,0.6740,0.1880}
  \definecolor{color6}{rgb}{0.3010,0.7450,0.9330}
  \definecolor{color7}{rgb}{0.6350,0.0780,0.1840}
  \pgfplotscreateplotcyclelist{matlab}{
    color1,every mark/.append style={solid},mark=*\\
    color2,every mark/.append style={solid},mark=square*\\
    color3,every mark/.append style={solid},mark=triangle*\\
    color4,every mark/.append style={solid},mark=halfsquare*\\
    color5,every mark/.append style={solid},mark=pentagon*\\
    color6,every mark/.append style={solid},mark=halfcircle*\\
    color7,every mark/.append style={solid,rotate=180},mark=halfdiamond*\\
    color1,every mark/.append style={solid},mark=diamond*\\
    color2,every mark/.append style={solid},mark=halfsquare right*\\
    color3,every mark/.append style={solid},mark=halfsquare left*\\
  }

  \tikzset{
    mark options={
      mark size=1.25,
    }
  }

  \pgfplotstableread[col sep=comma, sci, precision=10]{\datapath/\errorfile}{\datatable}
  \pgfplotstablegetrowsof{\datatable} 
  \pgfmathsetmacro{\numRows}{\pgfplotsretval-1}
  \pgfmathsetmacro{\numRowsM}{\numRows-1}
  \pgfmathsetmacro{\errMax}{0}
  \pgfmathsetmacro{\errMin}{10000}

  \pgfplotstablegetelem{0}{h}\of{\datatable}
  \pgfmathsetmacro{\hFirst}{\pgfplotsretval}
  \pgfplotstablegetelem{\numRows}{h}\of{\datatable}
  \pgfmathsetmacro{\hLast}{\pgfplotsretval}
  \pgfplotstablegetelem{\numRowsM}{h}\of{\datatable}
  \pgfmathsetmacro{\hLastM}{\pgfplotsretval}

  \begin{groupplot}[
    group style={
      group size= 2 by 1,
    },
    xlabel={\(h_{\max}\)},
    grid=major,
    cycle list name=matlab,
    width=\textwidth/2,
    xmode=log,
    ymode=log,
    ytick distance={10},
    height=4.75cm,
  ]
    \nextgroupplot[
      title={Relative \(L^2\)-errors},
      legend to name={CommonLegend_\errorfile},
      legend style={
        legend columns=-1,
        font=\footnotesize
      }
    ]
      \foreach \i/\j in \fieldList {
        \addplot table[x=h, y=\i_L_2, col sep=comma]{\datatable};
        \addlegendentryexpanded{\j} 
      }
      \addplot [dashed, domain=\hLast:\hFirst] {\firstOrderFactor*x^1};
      \addlegendentryexpanded{$\mathcal{O}(h)$}
      \addplot [dotted, domain=\hLast:\hFirst] {\secondOrderFactor*x^2};
      \addlegendentryexpanded{$\mathcal{O}(h^2)$}
    \nextgroupplot[
      title={Relative \(l^\infty\)-errors}
    ]
      \foreach \i/\j in \fieldList {
        \addplot table[x=h, y=\i_l_inf, col sep=comma]{\datatable};
      }
      \addplot [dashed, domain=\hLast:\hFirst] {\firstOrderFactor*x^1};
      \addplot [dotted, domain=\hLast:\hFirst] {\secondOrderFactor*x^2};

    \end{groupplot}
    \coordinate (c3) at ($(group c1r1)!.5!(group c2r1)$); \node[below] at (c3 |- current bounding box.south) {\pgfplotslegendfromname{CommonLegend_\errorfile}};
\end{tikzpicture}
  \caption{
      Relative errors using stabilized \(\mathbb{P}_2\) equal-order elements for the homogenous flow around a cylinder: All fields have second-order convergence rates in the \(L^2\)-norm. In the \(l^\infty\)-norm, reduced rates are observed for the \(\theta\)-field, but first-order convergence is guaranteed. Note the different \(h_{\max}\)-axis compared to \cref{fig:r13_1_coeffs_nosources_norot_inflow_p1p1p1p1p1_stab}.
    }\label{fig:r13_1_coeffs_nosources_norot_inflow_p2p2p2p2p2_stab}
    \Description{
      Relative errors using stabilized \(\mathbb{P}_2\) equal-order elements for the homogenous flow around a cylinder: All fields have second-order convergence rates in the \(L^2\)-norm. In the \(l^\infty\)-norm, reduced rates are observed for the \(\theta\)-field, but first-order convergence is guaranteed. Note the different \(h_{\max}\)-axis compared to \cref{r13_1_coeffs_nosources_norot_inflow_p1p1p1p1p1_stab}.
    }
\end{figure}
\par
The \cref{fig:r13_1_coeffs_nosources_norot_inflow_p2p2p2p2p2_stab_schematic} shows schematic results for this test case using \(h_{\max}=0.09\). The homogenous outer flow field enters the computational domain in \cref{sfig:r13_1_coeffs_nosources_norot_inflow_p2p2p2p2p2_stab_stress_schematic} in parallel. The shear-stress component \(\sigma_{xy}\) has a greater magnitude at areas where the flow field is parallel to the inner cylinder walls. As expected with the given set of heat boundary conditions, heat flux from the warm outer cylinder wall to the cold inner wall is present in \cref{sfig:r13_1_coeffs_nosources_norot_inflow_p2p2p2p2p2_stab_heat_schematic}. However, the flow advects the temperature field in the flow direction leading to a cold gas region only behind the cylinder. In contrast, the region before the cylinder is warm.
\begin{figure}[t]
  \begin{subfigure}[c]{0.49\textwidth}
    \includegraphics[width=\textwidth]{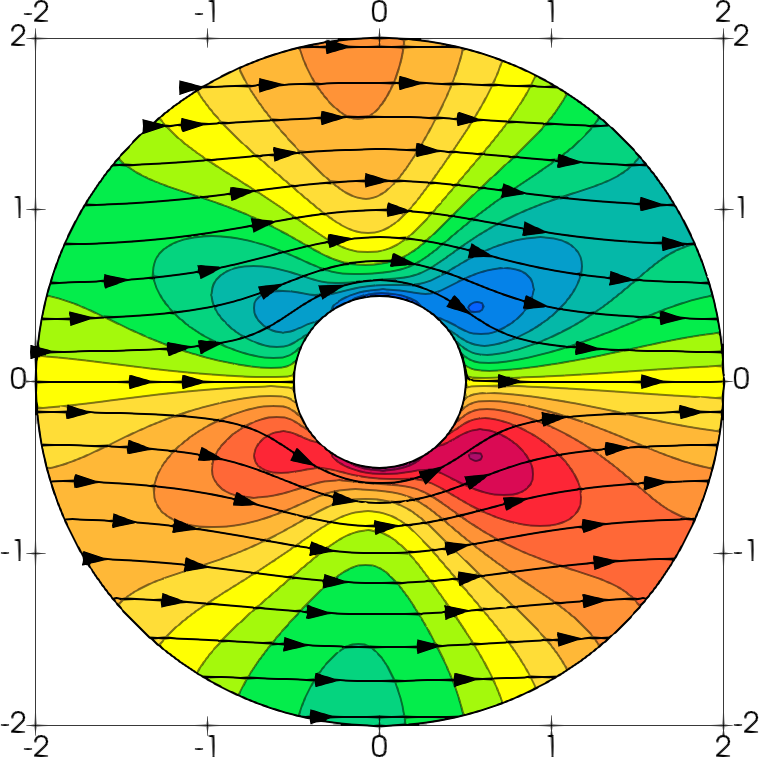}
    \subcaption{
      Shear stress \(\sigma_{xy}\), velocity streamlines \(u_{i}\).
    }\label{sfig:r13_1_coeffs_nosources_norot_inflow_p2p2p2p2p2_stab_stress_schematic}
  \end{subfigure}
  \hfill
  \begin{subfigure}[c]{0.49\textwidth}
    \includegraphics[width=\textwidth]{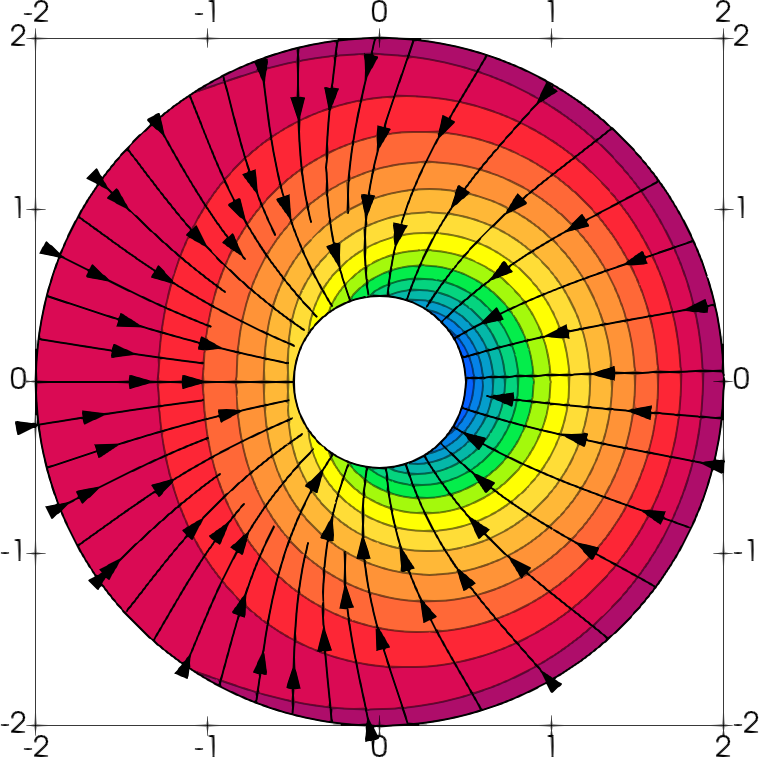}
    \subcaption{
      Temperature \(\theta\), heat flux streamlines \(s_{i}\).
    }\label{sfig:r13_1_coeffs_nosources_norot_inflow_p2p2p2p2p2_stab_heat_schematic}
  \end{subfigure}
  \caption{
    Schematic results of the homogenous flow around a cylinder for \(h_{\max}=0.047\): The flow past the cylinder in \cref{sfig:r13_1_coeffs_nosources_norot_inflow_p2p2p2p2p2_stab_stress_schematic} induces higher \(\abs{\sigma_{xy}}\) values above and below the cylinder. The flow direction is tangential to the inner cylinder walls. The temperature distribution \cref{sfig:r13_1_coeffs_nosources_norot_inflow_p2p2p2p2p2_stab_heat_schematic} reveals a colder region behind the cylinder due to the present flow field.
  }\label{fig:r13_1_coeffs_nosources_norot_inflow_p2p2p2p2p2_stab_schematic}
  \Description{
    Schematic results of the homogenous flow around a cylinder for \(h_{\max}=0.047\): The flow past the cylinder in \cref{sfig:r13_1_coeffs_nosources_norot_inflow_p2p2p2p2p2_stab_stress_schematic} induces higher \(\abs{\sigma_{xy}}\) values above and below the cylinder. The flow direction is tangential to the inner cylinder walls. The temperature distribution \cref{sfig:r13_1_coeffs_nosources_norot_inflow_p2p2p2p2p2_stab_heat_schematic} reveals a colder region behind the cylinder due to the present flow field.
  }
\end{figure}

\section{Application Cases}\label{s_applications}
To the test solver's non-equilibrium capabilities, we now discuss three application cases with typical rarefaction effects. The channel flow example in \cref{sec_channelFlow} shows the expected Knudsen paradox behavior. In \cref{sec_knudsenPump}, a temperature gradient at the domain walls induces a thermal transpiration flow without the presence of gravity. In \cref{sec_thermalEdgeFlow}, we compare to the existing literature by considering a thermally-induced edge flow. These three application cases would not be possible using a classical NSF solver and justify using the R13 equations to predict rarefaction effects in non-standard flow situations.

\subsection{Knudsen Paradox in a Channel Flow}\label{sec_channelFlow}
A classic example is a microchannel flow similar to the one-dimensional case discussed in~\cite{torrilhon2016modeling}. In our artificial two-dimensional setting, the flow domain \(\Omega \in \mathbb{R}^2\) is between two infinitely large plates. The domain length is \(L=4\), and the plate distance reads \(H=1\). Inside the domain \(\Omega\), a body force \(\te{b}={(1,0)}^T\) induces a flow in positive \(x\)-direction. We prescribe no pressure gradients with \(p^{\mathrm{w}}|_{\Gamma_i}=0\). No additional inflow or outflow velocities are assumed with \(u_n^{\mathrm{w}}|_{\Gamma_i}, u_t^{\mathrm{w}}|_{\Gamma_i}=0\). A uniform temperature \(\theta^{\mathrm{w}}|_{\Gamma_i}=1\) is applied to all boundaries. The upper and lower impermeable walls have \(\epsilon^{\mathrm{w}}|_{\Gamma_1,\Gamma_3}=10^{-3}\) while the in- and outflow walls are modeled with the parameter \(\epsilon^{\mathrm{w}}|_{\Gamma_2,\Gamma_4}=10^{3}\) to allow a velocity through these boundaries. The \cref{sfig:applicationsDomainChannel} presents the overall setup.
\begin{figure}[t]
  \begin{subfigure}[c]{0.54\textwidth}
    \newcommand{\stylefile}{figs/tikz/sketches/style.tex}



\begin{tikzpicture}[scale=1.5, rotate=0]

  \input{\stylefile}

  \pgfmathsetmacro{\height}{1};
  \pgfmathsetmacro{\sidelength}{4};
  \pgfmathsetmacro{\rOuter}{2};
  \pgfmathsetmacro{\solidThickness}{0.15};

  \draw [fRegion] (0,0) rectangle (\sidelength,\height);
  \draw node at ({0.5*\sidelength},{0.5*\height}) {\(\Omega\)};
  \draw node [above] at (0.5*\sidelength,0) {\(\Gamma_1\)};
  \draw node [left] at (\sidelength,0.5*\height) {\(\Gamma_2\)};
  \draw node [below] at (0.5*\sidelength,\height) {\(\Gamma_3\)};
  \draw node [right] at (0,0.5*\height) {\(\Gamma_4\)};

  \draw [bRegion] (0,\height) rectangle (\sidelength,\height+\solidThickness);
  \draw [bRegion] (0,-\solidThickness) rectangle (\sidelength,0);

  \pgfmathsetmacro{\dist}{0.3};
  \draw[dim] (0,-\dist) -- (\sidelength,-\dist);
  \draw node [below] at (0.5*\sidelength,-\dist) {\(L\)};
  \draw[dim] (0.75*\sidelength,0) -- (0.75*\sidelength,\height);
  \draw node [right] at (0.65*\sidelength,0.5*\height) {\(H\)};

  \pgfmathsetmacro{\velLength}{0.3};
  \pgfmathsetmacro{\velOffset}{0.1};
  \pgfmathsetmacro{\num}{4};
  \pgfmathsetmacro{\dist}{\height/\num};
  \foreach \yPos in {0,...,\num}{
    \draw[arrow,->] ({-(\velOffset+\velLength)},\yPos*\dist) -- ({-(\velOffset)},\yPos*\dist);
    \draw[arrow,->] ({\sidelength+(\velOffset)},\yPos*\dist) -- ({\sidelength+(\velOffset+\velLength)},\yPos*\dist);
  }

  \draw[arrow,->] (0.75,0.25) -- ({0.75+2*\velLength},0.25) node[above, midway] {\(\boldsymbol{b}\)};

\end{tikzpicture}

    \subcaption{
      Geometry.
    }\label{sfig:applicationsDomainChannel}
  \end{subfigure}
  \hfill
  \begin{subfigure}[c]{0.44\textwidth}
    \vspace{0.3cm}





\begin{tikzpicture}

  \definecolor{color1}{rgb}{0,    ,0.4470,0.7410}
  \definecolor{color2}{rgb}{0.8500,0.3250,0.0980}
  \definecolor{color3}{rgb}{0.9290,0.6940,0.1250}
  \definecolor{color4}{rgb}{0.4940,0.1840,0.5560}
  \definecolor{color5}{rgb}{0.4660,0.6740,0.1880}
  \definecolor{color6}{rgb}{0.3010,0.7450,0.9330}
  \definecolor{color7}{rgb}{0.6350,0.0780,0.1840}
  \pgfplotscreateplotcyclelist{matlab}{
    color1,mark=*\\
    color2,mark=square*\\
    color3,mark=triangle*\\
    color4,mark=halfsquare*\\
    color5,mark=pentagon*\\
    color6,mark=halfcircle*\\
    color7,mark=halfdiamond*\\
    color1,mark=otimes*\\
    color2,mark=diamond*\\
    color3,mark=halfsquare right*\\
    color4,mark=halfsquare left*\\
    color5,mark=square*\\
    color6,mark=square*\\
    color7,mark=square*\\
  }

  \tikzset{
    mark options={
      mark size=1.25,
    }
  }

  \pgfplotstableread[col sep=comma, sci, precision=10]{
    0.03125, 3.258048290000011
    0.0625, 1.9681735520439756
    0.125, 1.3755410480851753
    0.25, 1.1749283495650311
    0.5, 1.2471564053606263
    1.0, 1.5838568786925513
    2.0, 2.21831328293239
  }{\datatable}

  \begin{axis}[
    xmode=log,
    height=3.1cm,
    width=7cm, 
    xlabel={\(\text{Kn}\)},
    grid=major,
    cycle list name=matlab,
    legend style={at={(0.5,0.98)},anchor=north},
    legend entries={Mass flow \(\hat{J}_{\Gamma_2}\)},
  ]

  \addplot table[]{\datatable};
  \end{axis}

\end{tikzpicture}

    \subcaption{
      Knudsen paradox.
    }\label{sfig:applicationsKnudsenParadox}
  \end{subfigure}
  \caption{
    Geometry and Knudsen paradox in a channel flow. \cref{sfig:applicationsDomainChannel}: A dimensionless body force drives the flow in positive \(x\)-direction. The upper and lower boundaries are modeled as impermeable walls, while the inflow and outflow boundaries allow a mass flow. \cref{sfig:applicationsKnudsenParadox}: With increasing Knudsen number, the channel's mass flow first decreases to a minimum before increasing. Increasing the Knudsen number \(\Knud = \frac{\lambda}{H}\) either means diluting the gas or decreasing the channel width \(H\) for this particular context.
  }
  \Description{
    Geometry and Knudsen paradox in a channel flow. \cref{sfig:applicationsDomainChannel}: A dimensionless body force drives the flow in positive \(x\)-direction. The upper and lower boundaries are modeled as impermeable walls, while the inflow and outflow boundaries allow a mass flow. \cref{sfig:applicationsKnudsenParadox}: With increasing Knudsen number, the channel's mass flow first decreases to a minimum before increasing. Increasing the Knudsen number \(\Knud = \frac{\lambda}{H}\) either means diluting the gas or decreasing the channel width \(H\) for this particular context.
  }
\end{figure}
\par
The resulting computational mesh consists of 10712 uniform but unstructured triangles with 5517 nodes. We discretize the domain using \(\mathbb{P}_1\) equal-order finite elements and CIP stabilization with \(\delta_\theta=1,\delta_{\te{u}}=1, \delta_p=0.1\). The \cref{fig:applicationsResultsChannel} corresponds to a Knudsen number of \(\Knud=0.1\). In \cref{sfig:channelStress}, the flow field is almost parallel to the outer walls. The deviatoric stress component \(\sigma_{xy}\) has its maxima at both outer channel walls. In \cref{sfig:channelHeat}, the heat flux is nonzero, although no temperature gradient is applied.
\begin{figure}[t]
  \begin{subfigure}[c]{0.49\textwidth}
    \includegraphics[width=\textwidth]{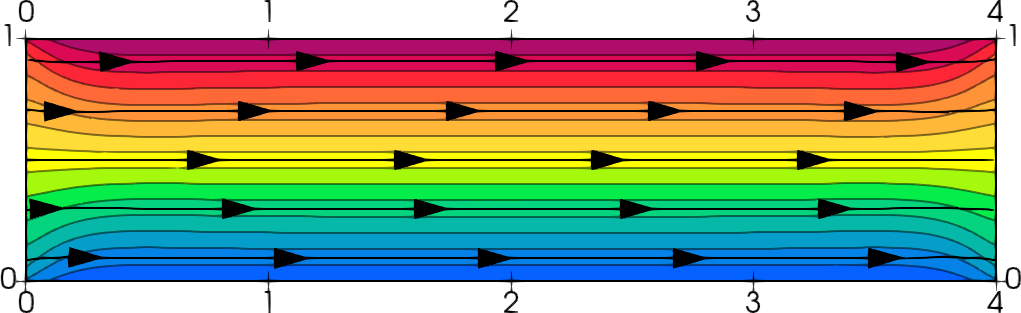}
    \subcaption{
      Shear stress \(\sigma_{xy}\), velocity streamlines \(u_{i}\).
    }\label{sfig:channelStress}
  \end{subfigure}
  \hfill
  \begin{subfigure}[c]{0.49\textwidth}
    \includegraphics[width=\textwidth]{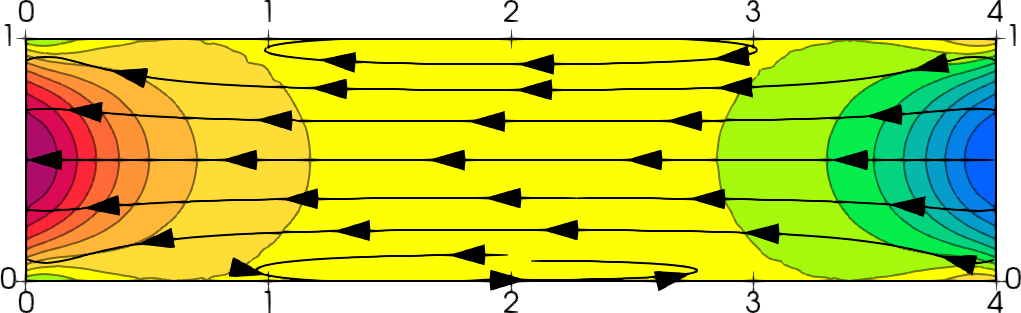}
    \subcaption{
      Temperature \(\theta\), heat flux streamlines \(s_{i}\).
    }\label{sfig:channelHeat}
  \end{subfigure}
  \caption{
    Schematic results of the channel flow application case for \(\Knud=0.1\): In \cref{sfig:channelStress}, the velocity field is almost parallel to the outer walls. The \cref{sfig:channelHeat} reveals a heat flux in the inflow direction.
  }\label{fig:applicationsResultsChannel}
  \Description{
    Schematic results of the channel flow application case for \(\Knud=0.1\): In \cref{sfig:channelStress}, the velocity field is almost parallel to the outer walls. The \cref{sfig:channelHeat} reveals a heat flux in the inflow direction.
  }
\end{figure}
\par
A parameter study for the Knudsen number further validates the solver capability to predict rarefaction effects using a range of \(\Knud \in [0.03,2.0]\) and the same overall problem setup. We define the dimensionless mass flow rate through the outflow boundary as
\begin{equation}
  \hat{J}_{\Gamma_2} = \int_{\Gamma_2} \te{u} \cdot \te{n} \dd l,
\end{equation}
similar to the one-dimensional considerations in~\cite{torrilhon2016modeling}. The mass flow rate reduces for increasing Knudsen numbers. However, the dimensionless mass flow rate has a minimum at about \(\Knud \approx 0.3\), followed by a subsequent increase. This phenomenon is known as the \textit{Knudsen paradox}, following~\cite{torrilhon2016modeling}, and measurements observed this effect, e.g., in~\cite{dongari2009pressure}. The \cref{sfig:applicationsKnudsenParadox} presents the relation between the dimensionless mass flow and the Knudsen number.

\subsection{Thermal Transpiration Flow in a Knudsen Pump}\label{sec_knudsenPump}
We observe another rarefaction effect in a \textit{Knudsen pump} test case, which is inspired by~\cite{westerkamp2014stabilization,aoki2007numerical,leontidis2014numerical}. Without a body force acting, a flow is solely induced by a temperature gradient at the outer walls. For the test case, we consider a racetrack-shaped geometry created by slicing a ring with inner and outer radii \(R_1\) and \(R_2\) into two halves and placing two rectangular connections between these elements. The two connection elements have side lengths \(2L\) and \(R_2-R_1\) with \(L=1\), \(R_1=1/2\), and \(R_2=2\). We prescribe a linear temperature profile at both boundaries with four control points, as presented in \cref{sfig:applicationsDomainPump}. The temperatures \(\theta_0=0.5\) and \(\theta_1=1.5\) define the initial temperature difference of \(\Delta \theta = 1\).
\begin{figure}
  \subcaptionbox{
    Knudsen pump geometry.\label{sfig:applicationsDomainPump}
  }[0.49\textwidth]{
    \newcommand{\stylefile}{figs/tikz/sketches/style.tex}


\begin{tikzpicture}[scale=1.0, rotate=0]

  \input{\stylefile}

  \pgfmathsetmacro{\rInner}{0.5};
  \pgfmathsetmacro{\rOuter}{2};
  \pgfmathsetmacro{\height}{1};
  \pgfmathsetmacro{\length}{1};
  \pgfmathsetmacro{\rOuter}{2};
  \pgfmathsetmacro{\solidThickness}{0.25};

  \draw [fRegion]
    (-\length,-\rOuter) rectangle (\length,\rOuter)
    (-\length,-\rOuter) arc (270:90:\rOuter)
    (\length,-\rOuter) arc (-90:90:\rOuter);
  \draw [fRegion]
    (-\length,-\rInner) rectangle (\length,\rInner)
    (-\length,-\rInner) arc (270:90:\rInner)
    (\length,-\rInner) arc (-90:90:\rInner);

  \draw [bRegion]
    (-\length,-\rOuter-\solidThickness) rectangle (\length,\rOuter+\solidThickness)
    (-\length,-\rOuter-\solidThickness) arc (270:90:\rOuter+\solidThickness)
    (\length,-\rOuter-\solidThickness) arc (-90:90:\rOuter+\solidThickness)
    (-\length,-\rOuter) rectangle (\length,\rOuter)
    (-\length,-\rOuter) arc (270:90:\rOuter)
    (\length,-\rOuter) arc (-90:90:\rOuter);
  \draw [bRegion]
    (-\length,-\rInner+\solidThickness) rectangle (\length,\rInner-\solidThickness)
    (-\length,-\rInner+\solidThickness) arc (270:90:\rInner-\solidThickness)
    (\length,-\rInner+\solidThickness) arc (-90:90:\rInner-\solidThickness)
    (-\length,-\rInner) rectangle (\length,\rInner)
    (-\length,-\rInner) arc (270:90:\rInner)
    (\length,-\rInner) arc (-90:90:\rInner);

  \draw[point] (\length,\rInner) circle (0.05) node [above right] {\(\theta_1\)};
  \draw[point] (-\length,\rInner) circle (0.05) node [above left] {\(\theta_0\)};
  \draw[point] (-\length,-\rInner) circle (0.05) node [below left] {\(\theta_1\)};
  \draw[point] (\length,-\rInner) circle (0.05) node [below right] {\(\theta_0\)};
  \draw[point] (\length,\rOuter) circle (0.05) node [below right] {\(\theta_1\)};
  \draw[point] (-\length,\rOuter) circle (0.05) node [below left] {\(\theta_0\)};
  \draw[point] (-\length,-\rOuter) circle (0.05) node [above left] {\(\theta_1\)};
  \draw[point] (\length,-\rOuter) circle (0.05) node [above right] {\(\theta_0\)};

  \draw[point] (0,0) circle (0.05);

  \draw node at ({\length+\rInner+0.5*(\rOuter-\rInner)},0) {\(\Omega\)};

  \draw node [right] at ({\length+\rInner},0) {\(\Gamma_1\)};
  \draw node [above] at (0,\rInner) {\(\Gamma_2\)};
  \draw node [left] at ({-(\length+\rInner)},0) {\(\Gamma_3\)};
  \draw node [below] at (0,-\rInner) {\(\Gamma_4\)};
  \draw node [left] at ({\length+\rOuter},0) {\(\Gamma_5\)};
  \draw node [below] at (0,\rOuter) {\(\Gamma_6\)};
  \draw node [right] at ({-(\length+\rOuter)},0) {\(\Gamma_7\)};
  \draw node [above] at (0,-\rOuter) {\(\Gamma_8\)};

  \draw node [above, fill=white, inner sep=1pt] at (0.5*\length,0) {\(L\)};
  \draw[dim] (0,0) -- (\length,0);

  \draw node [right, fill=white, inner sep=1pt] at (-0.65*\length,-0.5*\rInner) {\(R_1\)};
  \draw[dim] (-0.75*\length,0) -- (-0.75*\length,-\rInner);
  \draw node [right] at (-0.75*\length,0.5*\rOuter) {\(R_2\)};
  \draw[dim] (-0.75*\length,0) -- (-0.75*\length,\rOuter);

\end{tikzpicture}
  }
  \hfill
  \subcaptionbox{
    Knudsen pump discretization with \(h_{\max}=0.25\).\label{sfig:applicationsMeshPump}
  }[0.49\textwidth]{
    \includegraphics[width=\linewidth]{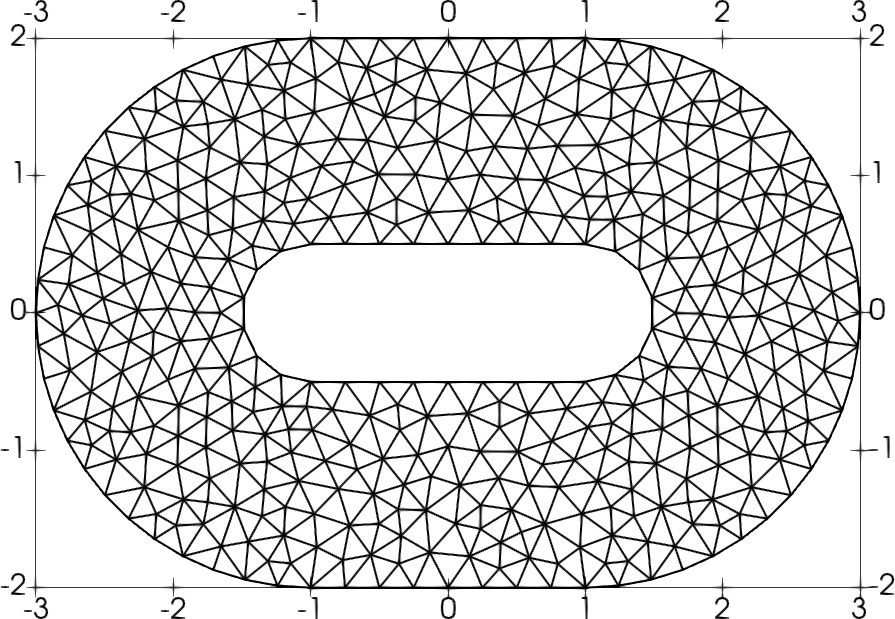}
  }
  \caption{
    Computational domain and spatial discretization for the Knudsen pump application case. \cref{sfig:applicationsDomainPump}: A temperature gradient at the outer walls drives the flow. All walls act as impermeable. The starting points of both half rings act as control points for the prescribed temperatures \(\theta_0\) and \(\theta_1\). In between these points, the temperature value follows a linear profile. \cref{sfig:applicationsMeshPump}: The maximum cell size \(h_{\max}\) characterizes the discretization with unstructured triangles.
  }
  \Description{
    Computational domain and spatial discretization for the Knudsen pump application case. \cref{sfig:applicationsDomainPump}: A temperature gradient at the outer walls drives the flow. All walls act as impermeable. The starting points of both half rings act as control points for the prescribed temperatures \(\theta_0\) and \(\theta_1\). In between these points, the temperature value follows a linear profile. \cref{sfig:applicationsMeshPump}: The maximum cell size \(h_{\max}\) characterizes the discretization with unstructured triangles.
  }
\end{figure}
\par
We want to prescribe a linear temperature profile for all the boundary paths between \(\theta_0\) and \(\theta_1\).
To derive the corresponding expressions for \(\theta^{\mathrm{w}}\), we use the polar angle function \(\operatorname{atan2}(x,y): \mathbb{R}^2 \rightarrow (-\pi,\pi]\). This function returns the polar angle, such that, e.g., \(\operatorname{atan2}(1,1) = \frac{\pi}{4}\). This function is available in most programming languages, in Python with switched arguments as ``\texttt{atan2(y,x)}''. The \cref{tab:applicationsKnudsenPumpTemperatureExpressions} shows the corresponding boundary expressions for the temperature at the inner wall. For the outer wall, the same expressions as for the inner wall are used (\(\theta^{\mathrm{}}|{_{\Gamma_5}} = \theta^{\mathrm{}}|{_{\Gamma_1}},\ldots\)). The remaining parameters \(u_n^{\mathrm{w}}\), \(u_t^{\mathrm{w}}\), and \(\epsilon^{\mathrm{w}}\), are all set to zero to model impermeable walls. Ee use \(\mathbb{P}_1\) equal-order elements and CIP stabilization with \(\delta_\theta=1, \delta_{\te{u}}=1,\delta_p=0.1\) for the numerical computation. 
\begin{table}[t]
  \centering
  \caption{Temperature boundary conditions for the Knudsen pump: The four boundary paths define the linear and continuous temperature profile for the outer walls.}\label{tab:applicationsKnudsenPumpTemperatureExpressions}
  \begin{tabular}{rcccc}
    \toprule
    \(\Gamma_i\) | &
    \(\Gamma_1\) &
    \(\Gamma_2\) &
    \(\Gamma_3\) &
    \(\Gamma_4\) \\
    \midrule
    \(\theta(x,y)|{_{\Gamma_i}}\) | &
    \(\frac{1}{2} \frac{2}{\pi} \operatorname{atan2}(x-1,y) + 1\) &
    \(\frac{1}{2} x + 1\) &
    \(- \frac{1}{2} \frac{2}{\pi} \operatorname{atan2}(1-x,y) + 1\) &
    \(-\frac{1}{2} x + 1\) \\
    \bottomrule
  \end{tabular}%
\end{table}
\par
To study the mesh sensitivity, we perform a series of computations on uniform unstructured grids with characterizing \(h_{\max}=1/2^i\) for \(i=4,\dots,7\). All meshes result from a complete re-meshing procedure implying additional refinement at curved boundaries. The \cref{sfig:applicationsMeshPump} presents a schematic mesh for \(h_{\max}=1/2^2\) and \cref{tab:applicationsPumpConvergenceTable} reports the convergence study results. No further change in the mean cross-section \(x\)-velocity \(3/2 \int_{-2}^{-1/2} (|u_x(y)|)|_{x=0} \dd y\) indicates the results' accuracy.
\par
The \cref{fig:applicationsResultsPump} presents the discrete solution for a Knudsen number of \(\Knud=0.1\) on the most refined mesh. We observe the counter-clockwise gas flow in \cref{sfig:PumpStress}, and the stress component \(\sigma_{xy}\) has higher values at the more curved parts of the geometry near the inner wall. The linearly applied temperature profile at the wall is visible in \cref{sfig:PumpHeat}. However, this temperature profile changes throughout the pump width due to diffusion. Intuitively, heat flux occurs in between warm and cold regions.
\begin{figure}
  \begin{subfigure}[c]{0.49\textwidth}
    \includegraphics[width=\textwidth]{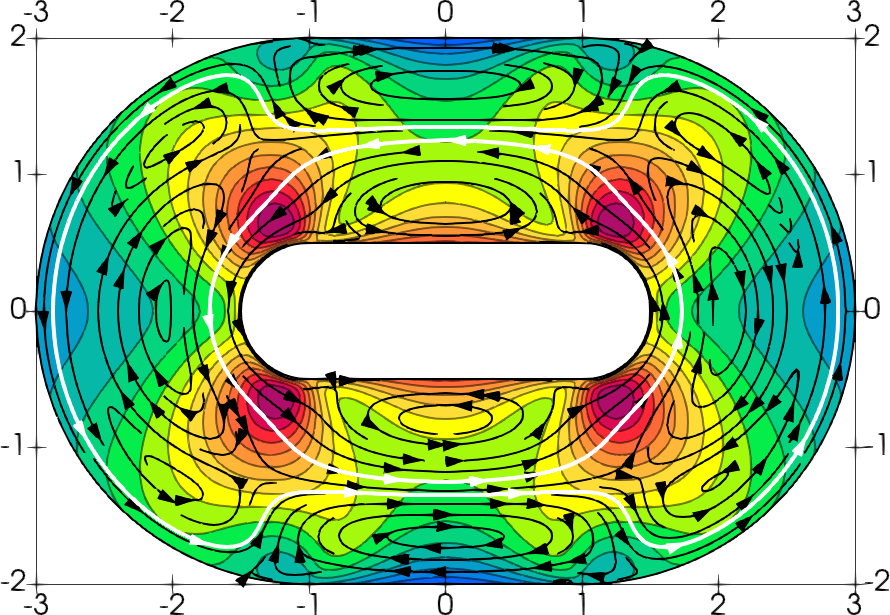}
    \subcaption{
      Shear stress \(\sigma_{xy}\), velocity streamlines \(u_{i}\).
    }\label{sfig:PumpStress}
  \end{subfigure}
  \hfill
  \begin{subfigure}[c]{0.49\textwidth}
    \includegraphics[width=\textwidth]{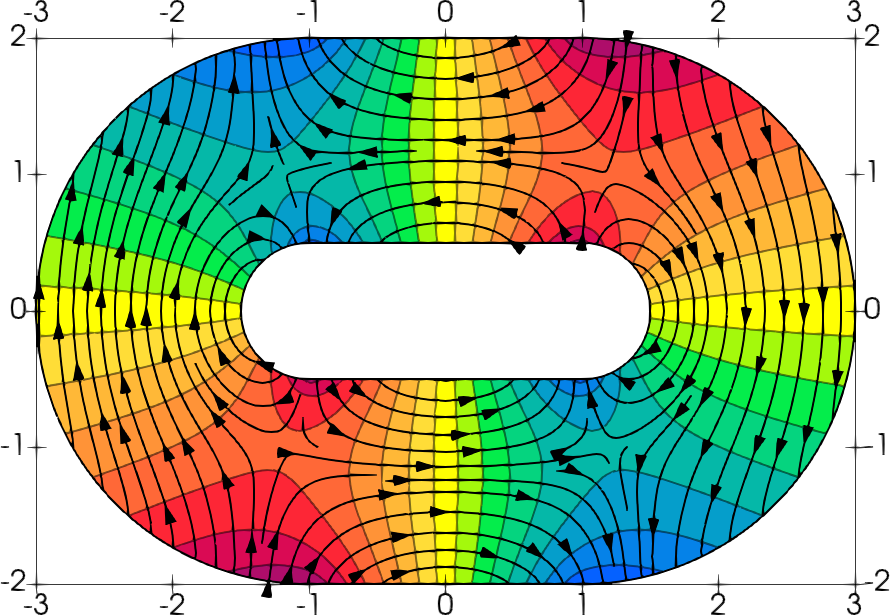}
    \subcaption{
      Temperature \(\theta\), heat flux streamlines \(s_{i}\).
    }\label{sfig:PumpHeat}
  \end{subfigure}
  \caption{
    Schematic results of the Knudsen pump application case for \(h_{\max}=1/2^{7}\): The gas flow rotates in the counter-clockwise direction, as observed in \cref{sfig:PumpStress}, similar to the results obtained in~\cite{westerkamp2014stabilization,aoki2007numerical}. The \cref{sfig:PumpHeat} shows the linear temperature gradient at both boundary walls, resulting in heat flux from warm to cold regions.
  }\label{fig:applicationsResultsPump}
  \Description{
    Schematic results of the Knudsen pump application case for \(h_{\max}=1/2^{7}\): The gas flow rotates in the counter-clockwise direction, as observed. The \cref{sfig:PumpHeat} shows the linear temperature gradient at both boundary walls, resulting in heat flux from warm to cold regions.
  }
\end{figure}
\begin{table}[t]
  \centering
  \caption{Summary of computations for the Knudsen pump, including: Number of triangles \(N_t\), number of nodes \(N_n\), maximum cell size \(h_{\max}\), mean cross-section \(x\)-velocity \(3/2 \int_{-2}^{-1/2} (|u_x(y)|)|_{x=0} \dd y\), wall time for FFC \(t_{\text{FFC}}\), wall time for system assembly \(t_{\text{a}}\), wall time for solution routine \(t_{\text{s}}\). The wall times were measured on an Amazon EC2 R5a instance using 64 vCPUs, 512 GiB memory, and MPI parallelization.}\label{tab:applicationsPumpConvergenceTable}
  \begin{filecontents*}{./data/knudsen_pump/convergence.csv}
12190  , 6286   , 0.0625    , 0.006874928077343663  , 27.39871311 , 0.39165329933166504 , 0.8428573608398438
48814  , 24788  , 0.03125   , 0.007023404183609014  , 0           , 0.9161298274993896  , 3.9293541908264167
192384 , 96952  , 0.015625  , 0.007061685143825088  , 0           , 2.5105490684509277  , 20.238901138305664
771972 , 387505 , 0.0078125 , 0.0070696479747835725 , 0           , 9.169977903366089   , 131.76481032371528
\end{filecontents*}
\pgfplotstabletypeset[
    col sep = comma,
    every head row/.style={before row=\toprule,after row=\midrule},
    every last row/.style={after row=\bottomrule},
    columns/0/.style={column name=\(N_{\text{t}}\)},
    columns/1/.style={column name=\(N_{\text{n}}\)},
    columns/2/.style={column name=\(h_{\max}\)},
    columns/3/.style={column name=\(3/2 \int_{-2}^{-1/2} (|u_x(y)|)|_{x=0} \dd y\)},
    columns/4/.style={column name=\(t_{\text{FFC}}\) [s]},
    columns/5/.style={column name=\(t_{\text{a}}\) [s]},
    columns/6/.style={column name=\(t_{\text{s}}\) [s]},
  ]
  {./data/knudsen_pump/convergence.csv}
\end{table}
\par
To further validate the obtained results, we also present velocity and temperature profiles at characteristic positions in \cref{fig:applicationsPumpPlotLines} for all considered meshes. The velocity profile at \(x=0\) shows a dominant flow in the middle of the domain. The temperature at \(x=1\) has its maximum not in the middle of the domain but around \(y \approx -1.1\).
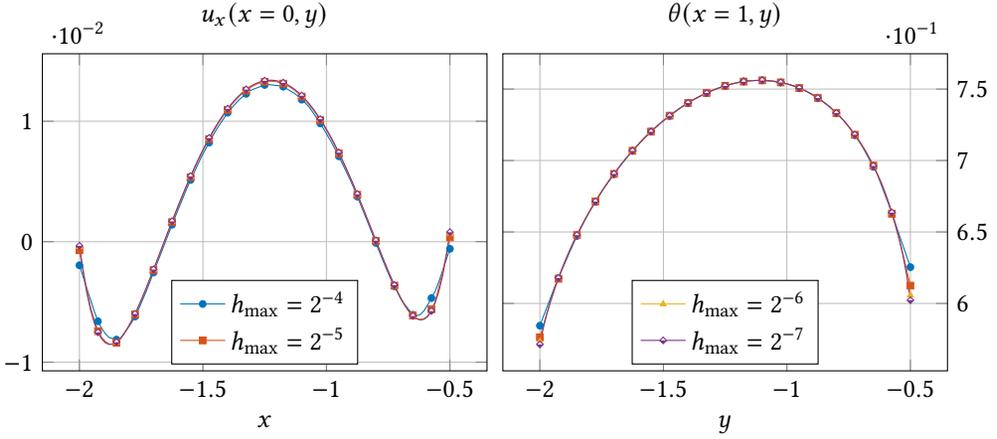
\begin{figure}[t]
  \newcommand{\datapath}{./data}%
  \centering




\begin{tikzpicture}

  \definecolor{color1}{rgb}{0,    ,0.4470,0.7410}
  \definecolor{color2}{rgb}{0.8500,0.3250,0.0980}
  \definecolor{color3}{rgb}{0.9290,0.6940,0.1250}
  \definecolor{color4}{rgb}{0.4940,0.1840,0.5560}
  \definecolor{color5}{rgb}{0.4660,0.6740,0.1880}
  \definecolor{color6}{rgb}{0.3010,0.7450,0.9330}
  \definecolor{color7}{rgb}{0.6350,0.0780,0.1840}
  \pgfplotscreateplotcyclelist{matlabmarkreordered}{
    color1,every mark/.append style={solid},mark=*\\
    color2,every mark/.append style={solid},mark=square*\\
    color3,every mark/.append style={solid},mark=triangle*\\
    color4,every mark/.append style={solid},mark=halfsquare*\\
    color5,every mark/.append style={solid},mark=pentagon*\\
    color6,every mark/.append style={solid},mark=halfcircle*\\
    color7,every mark/.append style={solid,rotate=180},mark=halfdiamond*\\
    color1,every mark/.append style={solid},mark=diamond*\\
    color2,every mark/.append style={solid},mark=halfsquare right*\\
    color3,every mark/.append style={solid},mark=halfsquare left*\\
  }

  \tikzset{
    mark options={
      mark size=1.25,
      mark repeat=5
    }
  }

  \begin{groupplot}[
    group style={
      group size=2 by 1,
      horizontal sep=6pt,
    },
    height=5.75cm,
    width=7.5cm,
    grid=major,
    cycle list name=matlabmarkreordered,
    legend pos=north east,
  ]
    \nextgroupplot[
      xlabel={\(x\)},
      legend entries={\(h_{\max}=2^{-4}\),\(h_{\max}=2^{-5}\),,},
      legend style={at={(0.5,0.02)},anchor=south},
      title={\(u_x(x=0,y)\)},
    ]
      \foreach \i in {0, ..., 3}
      {
        \pgfplotstableread[col sep=comma, sci, precision=10]{\datapath/knudsen_pump/ux-\i.csv}{\i}
        \addplot table[
          x=Points_1,
          y=u_0,
        ]{\i};
      }
    \nextgroupplot[
      xlabel={\(y\)},
      legend entries={,,\(h_{\max}=2^{-6}\),\(h_{\max}=2^{-7}\)},
      legend style={at={(0.5,0.02)},anchor=south},
      legend columns=1,
      title={\(\theta(x=1,y)\)},
      yticklabel pos=right,
      scaled y ticks=base 10:1,
    ]
      \foreach \i in {0, ..., 3}
      {
        \pgfplotstableread[col sep=comma, sci, precision=10]{\datapath/knudsen_pump/theta-\i.csv}{\i}
        \addplot table[
          x=Points_1,
          y=theta,
        ]{\i};
      }
  \end{groupplot}

\end{tikzpicture}

  \caption{
      Cross-section flow velocity and temperature profiles in a Knudsen pump for different spatial discretizations: The velocity profile shows a dominant counter-clockwise flow indicated by the positive mean velocity. We observe a significant stream in the domain's center and two opposite-directed streams near the boundaries. The temperature at \(x=1\) is maximal at \(y \approx -1.1 \ne 5/2\).
    }\label{fig:applicationsPumpPlotLines}
    \Description{
      Cross-section flow velocity and temperature profiles in a Knudsen pump for different spatial discretizations: The velocity profile shows a dominant counter-clockwise flow indicated by the positive mean velocity. We observe a significant stream in the domain's center and two opposite-directed streams near the boundaries. The temperature at \(x=1\) is maximal at \(y \approx -1.1 \ne 5/2\).
    }
\end{figure}

\subsection{Thermally-Induced Edge Flow}\label{sec_thermalEdgeFlow}
We finish the applications section by considering a geometry with a hot beam in a cold chamber, summarized in~\cref{fig:thermalEdgeflow,sfig:thermalEdgeflowGeometry}. We model both boundaries with \(\tilde{\chi}=1\) and apply a temperature difference with \(\theta|_{\Gamma_2}=1\) and \(\theta|_{\Gamma_1}=0\). This difference induces a flow field with \(\Knud=0.001\). Test cases like this are common for rarefied gas applications~\cite{su2020fast,su2020implicit}. For example, \citeauthor{su2020fast} used a different numerical method for the same geometry and a similar rarefied flow situation.
\begin{figure}
  \subcaptionbox{
    Geometry.\label{sfig:thermalEdgeflowGeometry}
  }[0.32\textwidth]{
    \newcommand{\stylefile}{figs/tikz/sketches/style.tex}



\begin{tikzpicture}[scale=0.52, rotate=0]

  \input{\stylefile}

  \pgfmathsetmacro{\sidelengthOuter}{8};
  \pgfmathsetmacro{\sidelengthInner}{2};
  \pgfmathsetmacro{\distInner}{1};
  \pgfmathsetmacro{\solidThickness}{0.29};

  \draw [fRegion] (0,0) rectangle (\sidelengthOuter,\sidelengthOuter);
  \draw node at ({0.5*\sidelengthOuter},{0.65*\sidelengthOuter}) {\(\Omega: \mathrm{Kn}\)};
  \draw node [below] at (0.5*\sidelengthOuter,\sidelengthOuter) {\(\Gamma_1: \theta^{\mathrm{w}}\)};

  \draw [fRegion] (\distInner,\distInner) rectangle (\distInner+\sidelengthInner,\distInner+\sidelengthInner);
  \draw node [above] at (\distInner+0.5*\sidelengthInner,\distInner+\sidelengthInner) {\(\Gamma_2: \theta^{\mathrm{w}}\)};

  \draw [bRegion] (-\solidThickness,-\solidThickness) rectangle (\sidelengthOuter+\solidThickness,0);
  \draw [bRegion] (-\solidThickness,0) rectangle (0,\sidelengthOuter);
  \draw [bRegion] (+\sidelengthOuter,0) rectangle (\sidelengthOuter+\solidThickness,\sidelengthOuter);
  \draw [bRegion] (0-\solidThickness,\sidelengthOuter) rectangle (\sidelengthOuter+\solidThickness,\sidelengthOuter+\solidThickness);

  \draw [bRegion] (\distInner,\distInner) rectangle (\distInner+\sidelengthInner,\distInner+\solidThickness);
  \draw [bRegion] (\distInner+\sidelengthInner-\solidThickness,\distInner) rectangle (\distInner+\sidelengthInner,\distInner+\sidelengthInner);
  \draw [bRegion] (\distInner,\distInner+\sidelengthInner-\solidThickness) rectangle (\distInner+\sidelengthInner,\distInner+\sidelengthInner);
  \draw [bRegion] (\distInner,\distInner) rectangle (\distInner+\solidThickness,\distInner+\sidelengthInner);

  \pgfmathsetmacro{\dist}{0.3};
  \draw[dim] (0,+\dist) -- (\sidelengthOuter,+\dist);
  \draw node [above] at (0.5*\sidelengthOuter,+\dist) {\(L\)};
  \draw[dim] (\distInner,\distInner+\dist) -- (\distInner+\sidelengthInner,\distInner+\dist);
  \draw node [above] at (\distInner+0.5*\sidelengthInner,\distInner+\dist) {\(l\)};

  \draw[dim] (0,\distInner+\sidelengthInner-\dist) -- (\distInner,\distInner+\sidelengthInner-\dist);
  \draw node [below] at (0.5*\distInner,\distInner+\sidelengthInner-\dist) {\(d\)};


\end{tikzpicture}

  }
  \hfill
  \subcaptionbox{
    Base mesh for \(s=0\).\label{sfig:thermalEdgeflowMesh}
  }[0.32\textwidth]{
    \includegraphics[width=\linewidth]{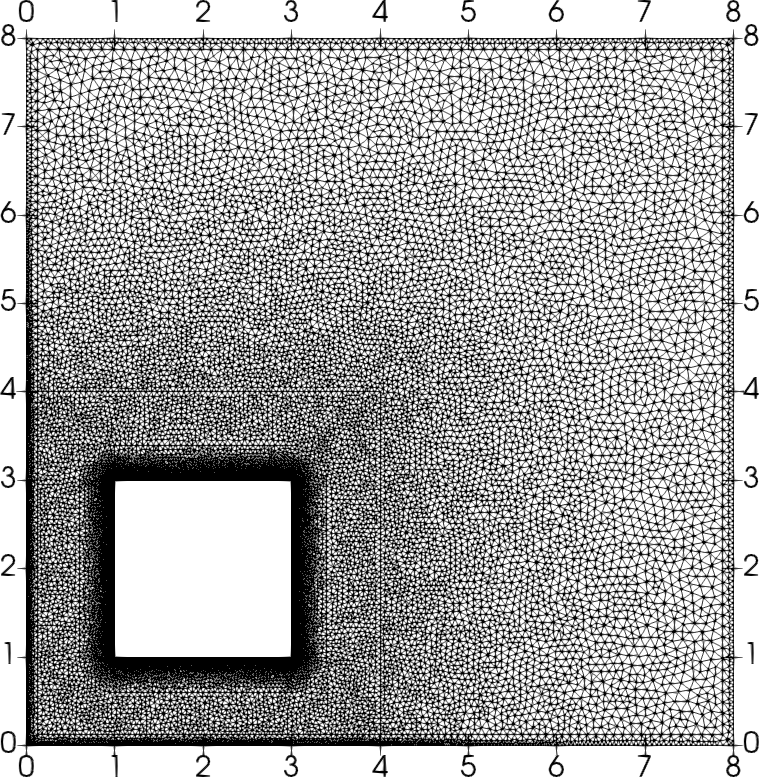}
  }
  \hfill
  \subcaptionbox{
    Schematic heat distribution.\label{sfig:thermalEdgeflowTemperatureHeatflux}
  }[0.32\textwidth]{
    \includegraphics[width=\linewidth]{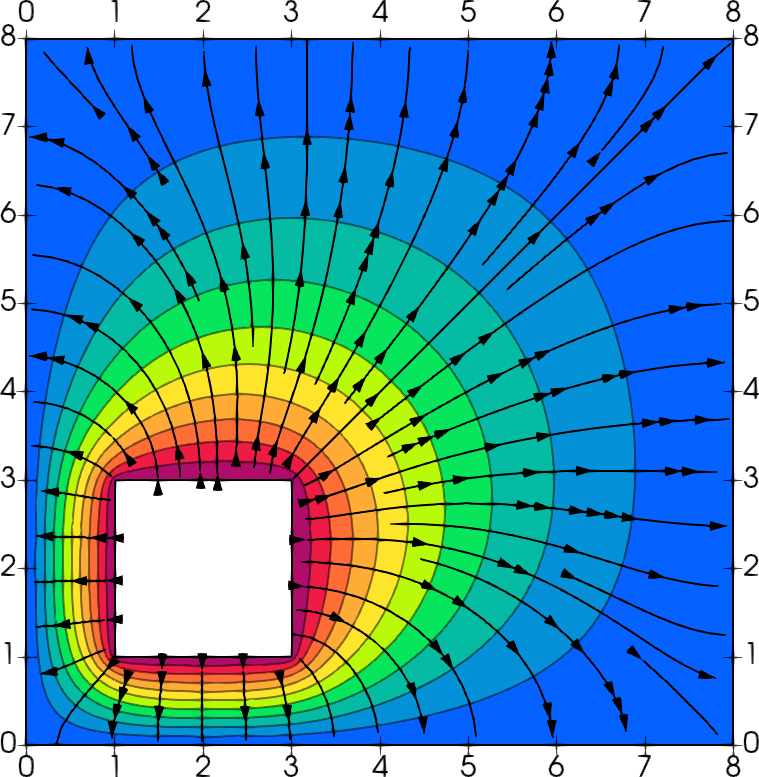}
  }
  \caption{
    Thermal edge flow for \(\Knud = 0.001\). \cref{sfig:thermalEdgeflowGeometry}: The geometry with \(L=8\), \(l=2\), \(d=1\) models a hot beam in a surrounding cold chamber. \cref{sfig:thermalEdgeflowMesh}: The base mesh with \(s=0\) has spatial refinement near the beam corners and near the outer walls. \cref{sfig:thermalEdgeflowTemperatureHeatflux}: The \(\theta\)-field and the \(\te{s}\)-streamlines show heat flux from the hot beam to the cold chamber walls. The temperature plot color consists of 10 uniform intervals with \(\theta \in [0,1]\).
  }
  \Description{
    Thermal edge flow for \(\Knud = 0.001\). \cref{sfig:thermalEdgeflowGeometry}: The geometry with \(L=8\), \(l=2\), \(d=1\) models a hot beam in a surrounding cold chamber. \cref{sfig:thermalEdgeflowMesh}: The base mesh with \(s=0\) has spatial refinement near the beam corners and near the outer walls. \cref{sfig:thermalEdgeflowTemperatureHeatflux}: The \(\theta\)-field and the \(\te{s}\)-streamlines show heat flux from the hot beam to the cold chamber walls. The temperature plot color consists of 10 uniform intervals with \(\theta \in [0,1]\).
  }\label{fig:thermalEdgeflow}
\end{figure}
\par
We perform a series of computations on unstructured triangle meshes with \(\mathbb{P}_1\) equal-order elements and CIP stabilization (\(\delta_\theta=1, \delta_{\te{u}}=1,\delta_p=0.1\)). The base mesh with \(s=0\) in~\cref{sfig:thermalEdgeflowMesh} applies refinement at the inner corners and the outer walls. All meshes result from a complete re-meshing procedure where we multiply the overall mesh resolution with the split factor \(2^{-s}\). In \cref{tab:applicationsEdgeflowConvergenceTable}, we summarize the corresponding mesh characterizations together with other computational statistics.
\begin{table}[t]
  \centering
  \caption{Summary of computations for the thermal edge flow, including: Split number \(s\), number of triangles \(N_t\), number of nodes \(N_n\), minimum cell size \(h_{\min}\), mean cross-section \(y\)-velocity \(1/8 \int_{0}^{8} (|u_y(x)|)|_{y=0.5} \dd x\), wall time for FFC \(t_{\text{FFC}}\), wall time for system assembly \(t_{\text{a}}\), wall time for solution routine \(t_{\text{s}}\). The wall times were measured on an Amazon EC2 R5a instance using 64 vCPUs, 512 GiB memory, and MPI parallelization.}\label{tab:applicationsEdgeflowConvergenceTable}
  \begin{filecontents*}{./data/thermal_edge_flow/convergence.csv}
0 , 58108   , 30572   , 0.00048828125   , 1.968857070660654e-06  , 21.92591119 , 0.9068424701690674 , 4.253207683563232
1 , 226758  , 116408  , 0.000244140625  , 1.8152936706753825e-06 , 0           , 2.8817651271820077 , 21.141472578048706
2 , 889196  , 450652  , 0.0001220703125 , 1.6218364406704659e-06 , 0           , 10.467012405395508 , 129.97580742835999
3 , 3533402 , 1778807 , 6.103515625e-5  , 1.5130567021718107e-06 , 0           , 40.105725049972534 , 801.9269967079163
\end{filecontents*}
\pgfplotstabletypeset[
    col sep = comma,
    every head row/.style={before row=\toprule,after row=\midrule},
    every last row/.style={after row=\bottomrule},
    columns/0/.style={column name=\(s\)},
    columns/1/.style={column name=\(N_{\text{t}}\)},
    columns/2/.style={column name=\(N_{\text{n}}\)},
    columns/3/.style={column name=\(h_{\min}\)},
    columns/4/.style={column name=\(1/8 \int_{0}^{8} (|u_y(x)|)|_{y=0.5} \dd x\)},
    columns/5/.style={column name=\(t_{\text{FFC}}\) [s]},
    columns/6/.style={column name=\(t_{\text{a}}\) [s]},
    columns/7/.style={column name=\(t_{\text{s}}\) [s]},
  ]
  {./data/thermal_edge_flow/convergence.csv}
\end{table}
\par
The \cref{sfig:thermalEdgeflowTemperatureHeatflux} shows the resulting heat distribution following Fourier's law. On the contrary, the flow field develops an edge flow rarefaction effect. Along the horizontal lines \(y=0.5\) and \(y=4.5\) in~\cref{fig:applicationsThermalEdgeFlowPlotLines}, we observe the same characteristic extremal points in the vertical velocity profile for all mesh resolutions. To further validate the results, we inspect the velocity's streamlines and magnitude in \cref{fig:edgeflowDetailedPlot}. We observe no significant change in the flow field with further refinement in \cref{sfig:uStreamlines2,sfig:uStreamlines3,sfig:uMagnitude2,sfig:uMagnitude3}. For all meshes, the flow field develops the characteristic eight vortices pattern. Altogether, the obtained results are in notable agreement with the high-resolution results of~\cite{su2020fast}. Together with the comparison to exact solutions in \cref{ss_convergenceStudy}, this fact underlines the method's correctness and validity.
\begin{figure}[t]
  \newcommand{\datapath}{./data}%
  \centering




\begin{tikzpicture}

  \definecolor{color1}{rgb}{0,    ,0.4470,0.7410}
  \definecolor{color2}{rgb}{0.8500,0.3250,0.0980}
  \definecolor{color3}{rgb}{0.9290,0.6940,0.1250}
  \definecolor{color4}{rgb}{0.4940,0.1840,0.5560}
  \definecolor{color5}{rgb}{0.4660,0.6740,0.1880}
  \definecolor{color6}{rgb}{0.3010,0.7450,0.9330}
  \definecolor{color7}{rgb}{0.6350,0.0780,0.1840}
  \pgfplotscreateplotcyclelist{matlabmarkreordered}{
    color1,every mark/.append style={solid},mark=*\\
    color2,every mark/.append style={solid},mark=square*\\
    color3,every mark/.append style={solid},mark=triangle*\\
    color4,every mark/.append style={solid},mark=halfsquare*\\
    color5,every mark/.append style={solid},mark=pentagon*\\
    color6,every mark/.append style={solid},mark=halfcircle*\\
    color7,every mark/.append style={solid,rotate=180},mark=halfdiamond*\\
    color1,every mark/.append style={solid},mark=diamond*\\
    color2,every mark/.append style={solid},mark=halfsquare right*\\
    color3,every mark/.append style={solid},mark=halfsquare left*\\
  }

  \tikzset{
    mark options={
      mark size=1.25,
      mark repeat=5
    }
  }

  \begin{groupplot}[
    group style={
      group size=2 by 1,
      horizontal sep=6pt,
    },
    height=5.75cm,
    width=7.5cm,
    grid=major,
    cycle list name=matlabmarkreordered,
    legend pos=north east,
  ]
    \nextgroupplot[
      xlabel={\(x\)},
      legend entries={\(s=0\),\(s=1\),,},
      legend pos=south east,
      legend columns=1,
      title={\(u_y(x,y=0.5)\)},
    ]
      \foreach \i in {0, ..., 3}
      {
        \pgfplotstableread[col sep=comma, sci, precision=10]{\datapath/thermal_edge_flow/0.5-\i.csv}{\i}
        \addplot table[
          x=Points_0,
          y=u_1,
        ]{\i};
      }
    \nextgroupplot[
      xlabel={\(x\)},
      legend entries={,,\(s=2\),\(s=3\)},
      legend pos=south east,
      legend columns=1,
      title={\(u_y(x,y=4.5)\)},
      yticklabel pos=right,
    ]
      \foreach \i in {0, ..., 3}
      {
        \pgfplotstableread[col sep=comma, sci, precision=10]{\datapath/thermal_edge_flow/4.5-\i.csv}{\i}
        \addplot table[
          x=Points_0,
          y=u_1,
        ]{\i};
      }
  \end{groupplot}

\end{tikzpicture}

  \caption{
    Flow velocity profiles for the thermal edge flow case on the different spatial discretization. The (\(y=0.5\))-profile reveals a local maximum at \(x=2\) surrounded by global minima at \(x \approx 1.3\) and \(x \approx 2.7\). Global maxima are at \(x \approx 0.8\) and \(x \approx 3.2\). The (\(y=4.5\))-profile reveals a characteristic maximum at the position \(x=2\), surrounded by minima in \(x \approx 0.55\) and \(x \approx 3.75\).
  }\label{fig:applicationsThermalEdgeFlowPlotLines}
  \Description{
    Flow velocity profiles for the thermal edge flow case on the different spatial discretization. The (\(y=0.5\))-profile reveals a local maximum at \(x=2\) surrounded by global minima at \(x \approx 1.3\) and \(x \approx 2.7\). Global maxima are at \(x \approx 0.8\) and \(x \approx 3.2\). The (\(y=4.5\))-profile reveals a characteristic maximum at the position \(x=2\), surrounded by minima in \(x \approx 0.55\) and \(x \approx 3.75\).
  }
\end{figure}
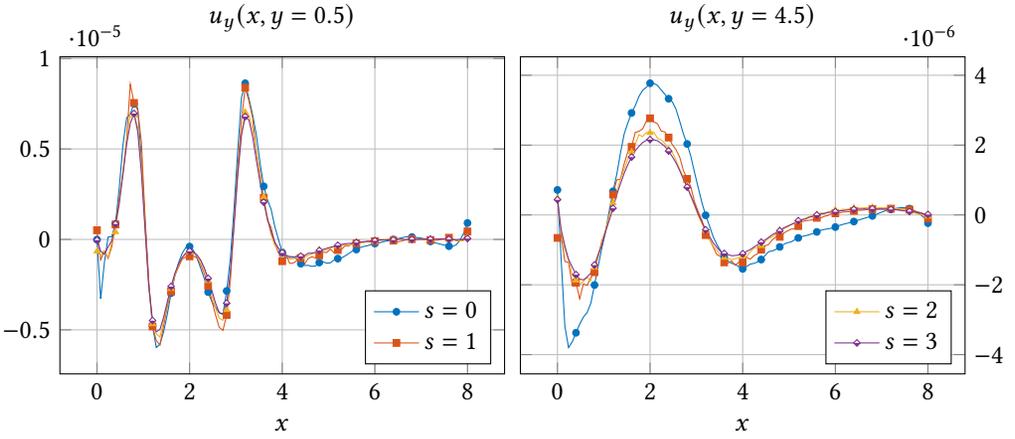
\begin{figure}
  \begin{subfigure}[c]{0.24\textwidth}
    \includegraphics[width=\textwidth]{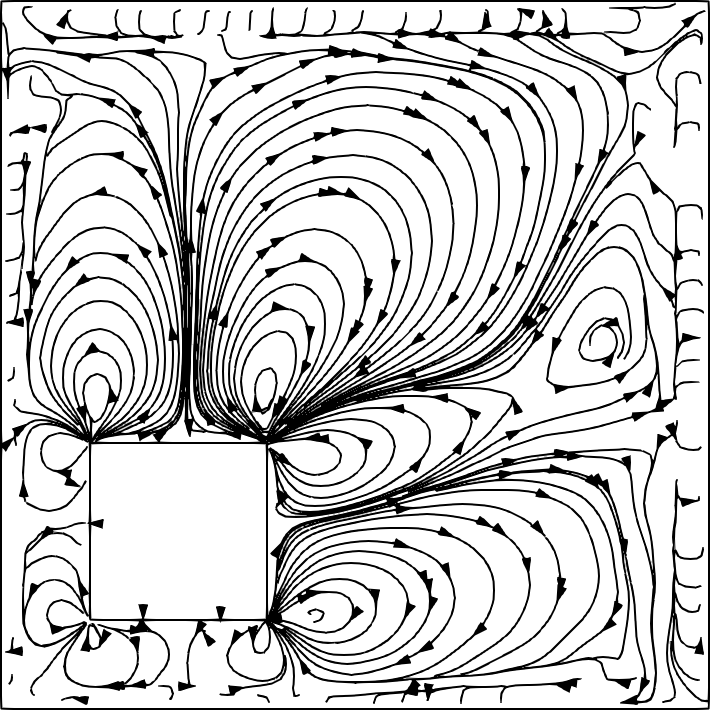}
    \subcaption{
      \(\te{u}\)-streamlines, \(s=0\).
    }\label{sfig:uStreamlines0}
  \end{subfigure}
  \hfill
  \begin{subfigure}[c]{0.24\textwidth}
    \includegraphics[width=\textwidth]{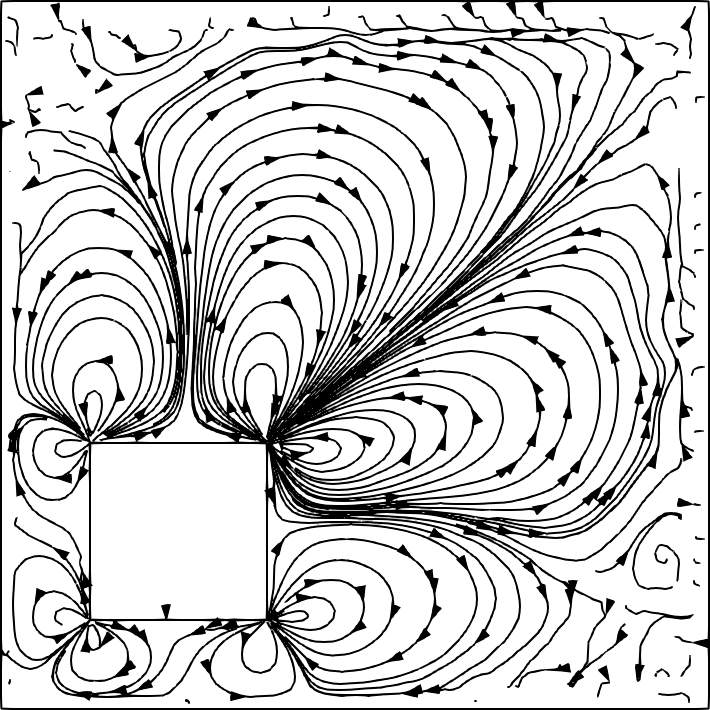}
    \subcaption{
      \(\te{u}\)-streamlines, \(s=1\).
    }\label{sfig:uStreamlines1}
  \end{subfigure}
  \hfill
  \begin{subfigure}[c]{0.24\textwidth}
    \includegraphics[width=\textwidth]{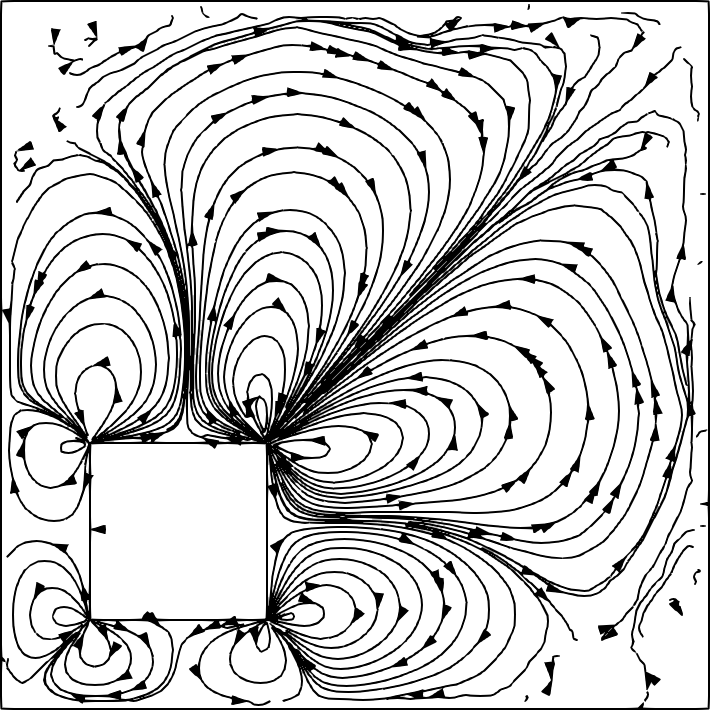}
    \subcaption{
      \(\te{u}\)-streamlines, \(s=2\).
    }\label{sfig:uStreamlines2}
  \end{subfigure}
  \hfill
  \begin{subfigure}[c]{0.24\textwidth}
    \includegraphics[width=\textwidth]{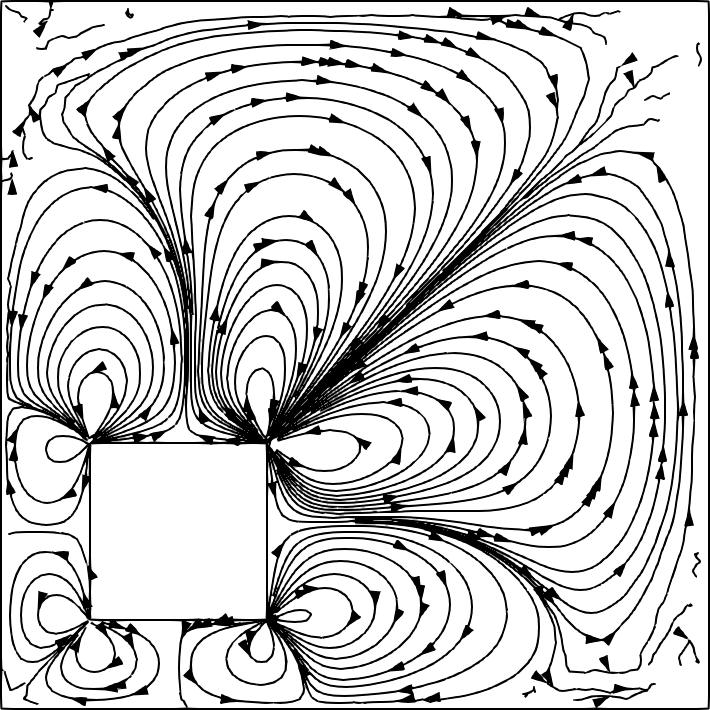}
    \subcaption{
      \(\te{u}\)-streamlines, \(s=3\).
    }\label{sfig:uStreamlines3}
  \end{subfigure}
  \begin{subfigure}[c]{0.24\textwidth}
    \includegraphics[width=\textwidth]{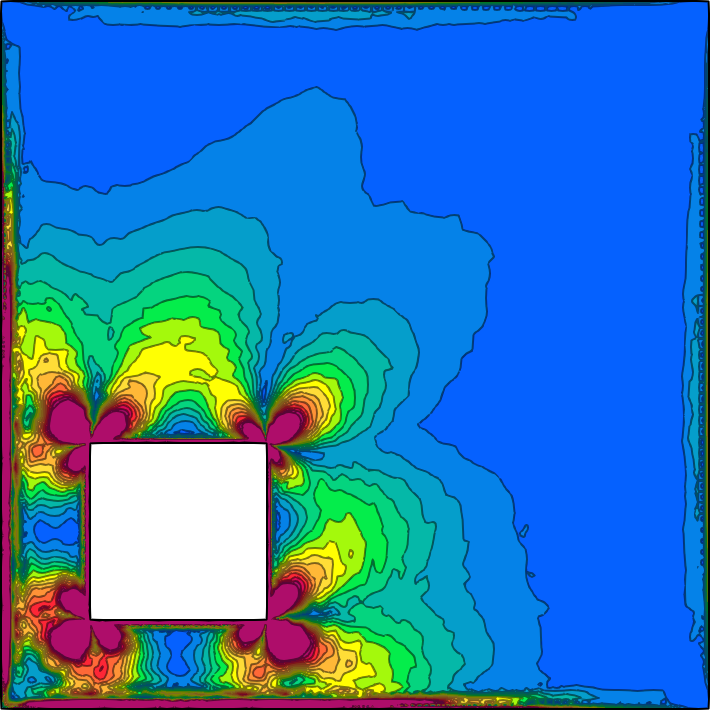}
    \subcaption{
      \(|\te{u}|\)-magnitude, \(s=0\).
    }\label{sfig:uMagnitude0}
  \end{subfigure}
  \hfill
  \begin{subfigure}[c]{0.24\textwidth}
    \includegraphics[width=\textwidth]{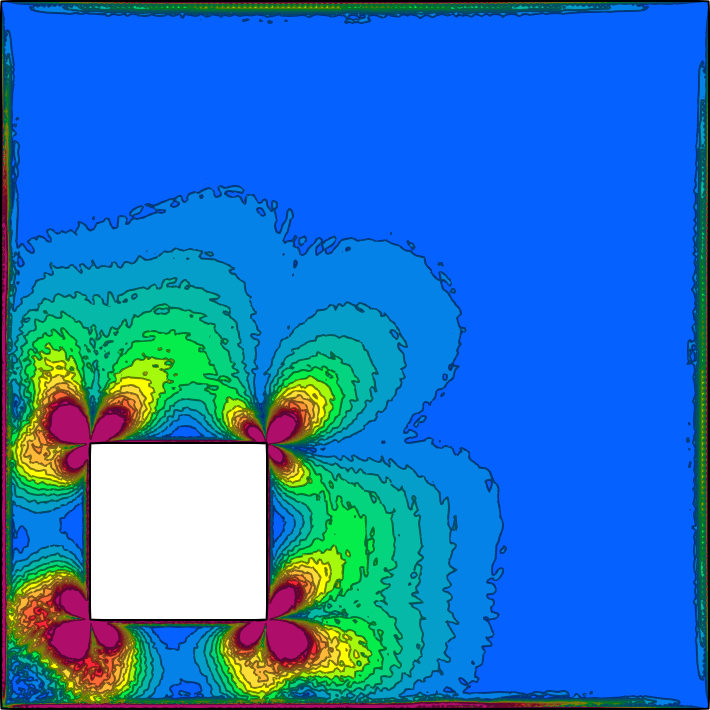}
    \subcaption{
      \(|\te{u}|\)-magnitude, \(s=1\).
    }\label{sfig:uMagnitude1}
  \end{subfigure}
  \hfill
  \begin{subfigure}[c]{0.24\textwidth}
    \includegraphics[width=\textwidth]{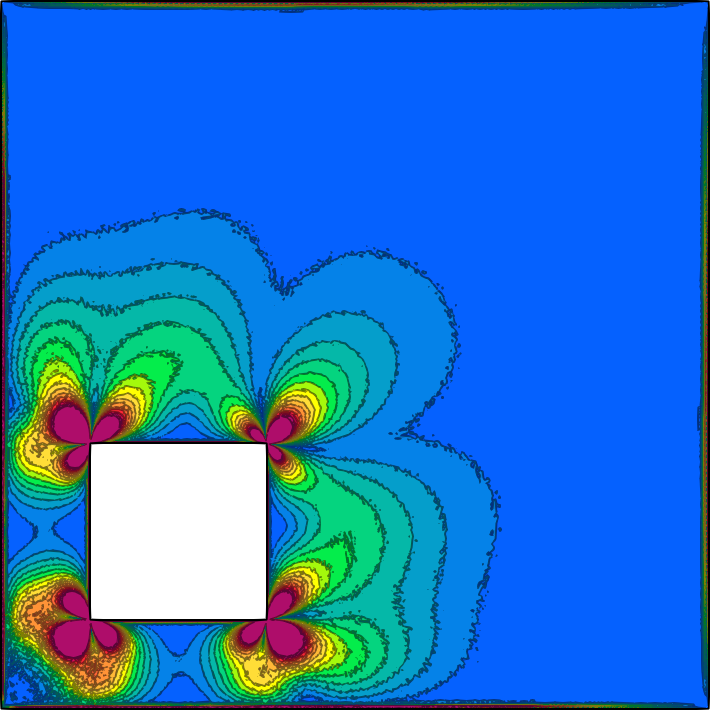}
    \subcaption{
      \(|\te{u}|\)-magnitude, \(s=2\).
    }\label{sfig:uMagnitude2}
  \end{subfigure}
  \hfill
  \begin{subfigure}[c]{0.24\textwidth}
    \includegraphics[width=\textwidth]{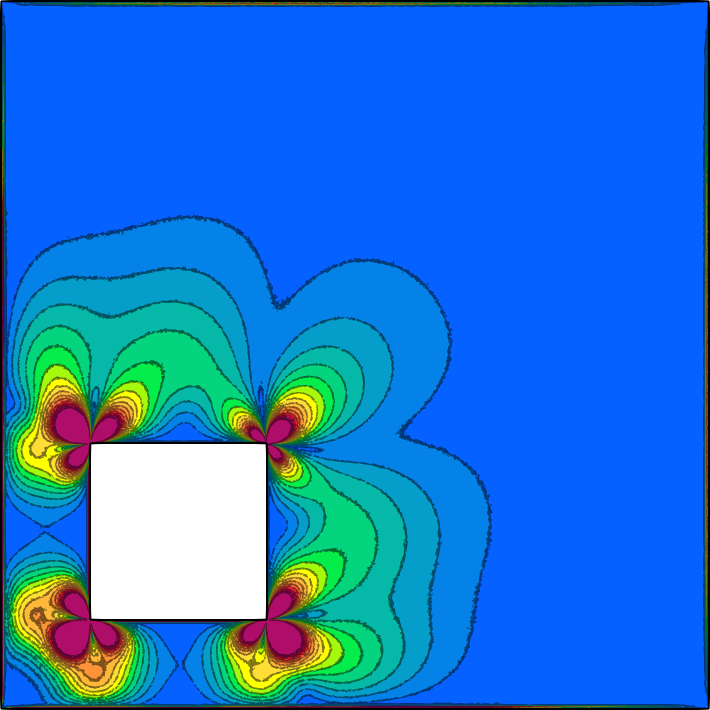}
    \subcaption{
      \(|\te{u}|\)-magnitude, \(s=3\).
    }\label{sfig:uMagnitude3}
  \end{subfigure}
  \caption{
    Detailed visualization of the flow field for the thermal edge flow case on the different spatial discretization: We observe no significant changes in the flow field between the two most refined meshes. \cref{sfig:uStreamlines0,sfig:uStreamlines1,sfig:uStreamlines2,sfig:uStreamlines3}: The velocity streamlines show the characteristic eight vortices flow pattern in agreement with~\cite{su2020fast}. \cref{sfig:uMagnitude0,sfig:uMagnitude1,sfig:uMagnitude2,sfig:uMagnitude3}: Significant peaks in the velocity develop in the vicinity of the beam corners. For better visualization, the color scheme of the \(|\te{u}|\)-magnitude is clipped at \(|\te{u}|=10^{-5}\).
    }\label{fig:edgeflowDetailedPlot}
  \Description{
    Detailed visualization of the flow field for the thermal edge flow case on the different spatial discretization: We observe no significant changes in the flow field between the two most refined meshes. \cref{sfig:uStreamlines0,sfig:uStreamlines1,sfig:uStreamlines2,sfig:uStreamlines3}: The velocity streamlines show the characteristic eight vortices flow pattern in agreement with~\cite{su2020fast}. \cref{sfig:uMagnitude0,sfig:uMagnitude1,sfig:uMagnitude2,sfig:uMagnitude3}: Significant peaks in the velocity develop in the vicinity of the beam corners. For better visualization, the color scheme of the \(|\te{u}|\)-magnitude is clipped at \(|\te{u}|=10^{-5}\).
  }
\end{figure}

\section{Conclusion and Outlook}\label{s_conclusion}
In this work, we presented a mixed finite element solver for the linear R13 equations. The solver implementation, provided publicly in~\cite{theisen2020fenicsr13Zenodo}, uses the tensor-valued description of the model equations. The weak form's derivation revealed the need for differential operators of variables with a tensor rank above two. Using the generality of FEniCS (and its underlying form language) allowed us to implement these operators conveniently in index notation. This abstraction level leads to an almost one-to-one correspondence between the mathematical formulation and its corresponding implementation, improving the source code's readability and maintainability.
\par
Furthermore, a convergence study showed the validity of the proposed method for a stable 5-tuple of elements and a stabilized equal-order combination. Application cases justified using the R13 equation over the traditional Navier--Stokes--Fourier models due to the capabilities to predict rarefaction effects for non-equilibrium gas flows.
\paragraph{UFL Discussion and Outlook}
This work showed the successful application of FEniCS's UFL to tensor-valued model equations. More general use cases involving equations with tensor rank \(r>2\) are possible using the ``\texttt{TensorElement}'' class of~\cite{fenics2020uflRepo} with the ``\texttt{shape=(d1,..,dr)}'' option. In this case, one must implement the required operators using the summation convention analogously to \cref{s_implementationAndValidation}. While we used the CIP surface stabilization, UFL also allows for residual-based volume stabilization, demonstrated, e.g., for the GLS method, in~\cite{helanow2018stabilized}. 
\par
Limitations of UFL include the FFC time \(t_{\text{FFC}}\) of order \(\mathcal{O}(1)\) (as reported in \cref{tab:applicationsPumpConvergenceTable,tab:applicationsEdgeflowConvergenceTable}), which is required once to compile the weak form. This compilation process also introduces another layer in the overall simulation pipeline, which might increase the debugging complexity. While being part of the FEniCS framework, it is also unclear whether one can use the UFL separately in other contexts without connection to FEniCS or the finite element method.
\par
However, for our particular use case, the generalization capabilities outweigh the possible limitations. Especially for quick prototyping of new model equations and their corresponding discretization, the UFL seems to be a practical choice. This fact holds in particular also for other moment models, e.g.~\cite{berghoff2020massively,koellermeier2017numerical,koellermeier2020spline}. Therefore, future work in rarefied gas applications could consider more complicated gas configurations~\cite{sarna2020moment,pavic2013maximum}, allow for more general gas molecules~\cite{cai2020regularized} or use different closures~\cite{mcdonald2013affordable}.
\begin{acks}
  The authors acknowledge financial support by the German Research Foundation (DFG) under grant IRTG 2379. The authors acknowledge the anonymous reviewers for their helpful comments.
\end{acks}

\bibliographystyle{ACM-Reference-Format}
\bibliography{references.bib}









\appendix

\section{Derivation of the Variational Formulation}
In the following, we derive the mixed variational formulation of \cref{ss_weakform} for all evolution equations \cref{eq_balance_mass,eq_balance_momentum,eq_balance_energy,eq_balance_heatflux,eq_balance_stress} separately. Reordering of the boundary conditions includes the addition of \cref{eq_bc_sigmant,eq_bc_rnt} to form \cref{eq_bcnew_rnt}, the addition of \cref{eq_bc_sn,eq_bc_mnnn} to form \cref{eq_bcnew_mnnn}, and a reordering/scaling of \cref{eq_bc_sigmant,eq_bc_sn} to obtain \cref{eq_bcnew_sigmant,eq_bcnew_sn} yields
\begin{align}
  u_n &= \epsilon^\mathrm{w} \tilde{\chi} \left( (p-p^\mathrm{w}) + \sigma_{nn} \right) + u_n^{\mathrm{w}}, \label{eq_bcnew_un}
  \\
  \frac{1}{\tilde{\chi}} \sigma_{nt} + u_t^{\mathrm{w}} - \frac{1}{5} s_t &= u_t + m_{nnt}, \label{eq_bcnew_sigmant}
  \\
  R_{nt} &= \tilde{\chi} \frac{12}{5} s_t - \sigma_{nt},\label{eq_bcnew_rnt}
  \\
  \frac{1}{2} \frac{1}{\tilde{\chi}} s_n + \theta^{\mathrm{w}} &= \theta + \frac{1}{4} \sigma_{nn} + \frac{1}{5} R_{nn} + \frac{1}{15} \Delta,\label{eq_bcnew_sn}
  \\
  \frac{3}{4} m_{nnn} &= \tilde{\chi} \frac{9}{8} \sigma_{nn} - \frac{3}{20} s_n,\label{eq_bcnew_mnnn}
  \\
  \left( \frac{1}{2} m_{nnn} + m_{ntt} \right) &= \tilde{\chi} \left( \frac{1}{2} \sigma_{nn} + \sigma_{tt} \right), \label{eq_bcnew_05mnnnmnnt}
\end{align}

\subsection{Heat Flux Balance}\label{ss_derviationHeatflux}
We start by testing the heat flux balance \cref{eq_balance_heatflux} with \(\te{r}\) and apply integration by parts to the terms
\begin{align}
  \int_\Omega \left(\te{\nabla} \te{\cdot} \tee{\sigma}\right) \te{\cdot} \te{r} \dd \te{x}
  &=
  - \int_\Omega \tee{\sigma} : \nabla \te{r} \dd \te{x} + \int_{\Gamma} \left( \tee{\sigma} \cdot \te{n} \right) \cdot \te{r} \dd l
  ,
  \\
  \frac{1}{2} \int_\Omega \left(\te{\nabla} \te{\cdot}\tee{R}\right)  \te{\cdot} \te{r} \dd \te{x}
  &=
  - \frac{1}{2}  \int_\Omega \tee{R} : \nabla \te{r} \dd \te{x} + \frac{1}{2}  \int_{\Gamma} \left( \tee{R} \cdot \te{n} \right) \cdot \te{r} \dd l
  ,
  \\
  \frac{5}{2} \int_\Omega \left(\te{\nabla} \theta\right) \te{\cdot} \te{r} \dd \te{x}
  &=
  - \frac{5}{2} \int_\Omega \theta \left( \nabla \cdot \te{r} \right) \dd \te{x} + \frac{5}{2} \int_{\Gamma} \theta \left( \te{r} \cdot \te{n} \right) \dd l
  ,
  \\
  \frac{1}{6} \int_\Omega \left(\te{\nabla} \Delta\right) \te{\cdot} \te{r} \dd \te{x}
  &=
  - \frac{1}{6} \int_\Omega \Delta \left( \nabla \cdot \te{r} \right) \dd \te{x} + \frac{1}{6} \int_{\Gamma} \Delta \left( \te{r} \cdot \te{n} \right) \dd l
  .
\end{align}
To proceed further, we expand all boundary integrals. In fact, for two spatial dimensions and regarding a local normal/tangential (characterized through \(\te{n},\te{t}\)) coordinate system along the boundary path, it holds that
\begin{align}
  \left( \tee{\sigma} \cdot \te{n} \right) \cdot \te{r} &= \sigma_{nn} r_n + \sigma_{nt} r_t,
  &
  \left( \tee{R} \cdot \te{n} \right) \cdot \te{r} &= R_{nn} r_n + R_{nt} r_t,
  \\
  \theta \left( \te{r} \cdot \te{n} \right) &= \theta r_n,
  &
  \Delta \left( \te{r} \cdot \te{n} \right) &= \Delta r_n.
\end{align}
The resulting equation for the left-hand side is normalized with the factor \((2/5)\) and reads
\begin{align}
  &
  - \frac{2}{5} \int_\Omega \tee{\sigma} : \nabla \te{r} \dd \te{x}
  - \frac{1}{5}  \int_\Omega \tee{R} : \nabla \te{r} \dd \te{x}
  - \int_\Omega \theta \left( \nabla \cdot \te{r} \right) \dd \te{x}
  - \frac{1}{15} \int_\Omega \Delta \left( \nabla \cdot \te{r} \right) \dd \te{x}
  \nonumber
  \\
  &
  \frac25 \int_{\Gamma} \left( \sigma_{nn} r_n + \sigma_{nt} r_t \right) \dd l
  + \frac{1}{5}  \int_{\Gamma} \left( R_{nn} r_n + R_{nt} r_t \right) \dd l
  + \int_{\Gamma} \theta r_n \dd l
  + \frac{1}{15} \int_{\Gamma} \Delta r_n \dd l
  ,
\end{align}
to insert the boundary conditions \cref{eq_bcnew_sn,eq_bcnew_rnt} together with the closure expression \cref{eq_closure_r,eq_closure_delta}. The final weak form of the heat balance therefore reads
\begin{align}
  0=
  &
  - \frac{2}{5} \int_\Omega \tee{\sigma} : \nabla \te{r} \dd \te{x}
  + \frac{24}{25} \Knud \int_\Omega {(\te{\nabla} \te{s})}_{\text{stf}} : \nabla \te{r} \dd \te{x}
  - \int_\Omega \theta \left( \nabla \cdot \te{r} \right) \dd \te{x}
  \nonumber
  \\
  &
  + \frac{4}{5} \Knud \int_\Omega \left( \nabla \cdot \te{s} \right) \left( \nabla \cdot \te{r} \right) \dd \te{x}
  + \frac{3}{20} \int_{\Gamma} \sigma_{nn} r_n
  + \frac{1}{5} \int_{\Gamma} \sigma_{nt} r_t \dd l
  + \frac{12}{25} \tilde{\chi} \int_{\Gamma} s_t r_t \dd l
  \nonumber
  \\
  &
  + \frac{1}{2} \frac{1}{\tilde{\chi}} \int_\Gamma s_n r_n \dd l
  + \int_\Gamma \theta^{\text{w}} r_n \dd l
  + \frac{4}{15} \frac{1}{\Knud} \int_\Omega \te{s} \te{\cdot} \te{r} \dd \te{x}
  ,
\end{align}
which equivalently, using the sub-functionals \cref{eq_subf_a,eq_subf_b,eq_subf_c,eq_subf_l1}, reads as
\begin{equation}
  a(\te{s},\te{r})
  -
  b(\theta, \te{r})
  -
  c(\te{r},\tee{\sigma})
  =
  l_1(\te{r})
  .
  \label{eq_subf_line_heatflux}
\end{equation}
Note that \cref{eq_subf_line_heatflux} is the result defining the symmetric and trace-free operator in three dimensions \cref{eq:heatStfModification} together with the orthogonality principle of the additive tensor decomposition into the symmetric and the skew-symmetric part as
\begin{align}
  \int_\Omega {(\te{\nabla} \te{s})}_{\text{stf}} : \nabla \te{r} \dd \te{x}
  &=
  \int_\Omega \text{sym}(\te{\nabla} \te{s}) : \text{sym}(\te{\nabla} \te{r}) \dd \te{x} - \frac{1}{3} \int_\Omega \left( \trace{(\te{\nabla} \te{s})} \tee{I} \right) \te{:} \te{\nabla} \te{r
  } \dd \te{x}
  \nonumber
  \\
  &=
  \int_\Omega \text{sym}(\te{\nabla} \te{s}) : \text{sym}(\te{\nabla} \te{r}) \dd \te{x} - \frac{1}{3} \int_\Omega \left( \te{\nabla} \te{\cdot} \te{s} \right) \te{\cdot} \left( \te{\nabla} \te{\cdot} \te{r} \right) \dd \te{x}
  .
\end{align}

\subsection{Energy Balance}\label{ss_derviationEnergy}
The energy equation \cref{eq_balance_energy} needs special treatment already in the strong form. To have an (anti-) symmetric system, as we see later on, we eliminate the velocity divergence utilizing the continuity equation \cref{eq_balance_mass}. This step is not necessary when using the variables \((\rho,\theta)\) instead of \((p, \theta)\). A subsequent testing with the scalar test function \(\kappa\) yields
\begin{equation}
  \int_\Omega \kappa \left( \te{\nabla} \te{\cdot} \te{s} \right) \dd \te{x}
  =
  \int_\Omega \left( r - \dot{m} \right) \kappa \dd \te{x}
  ,
\end{equation}
which equivalently, using the sub-functionals \cref{eq_subf_b,eq_subf_l2}, reads as
\begin{equation}
  b(\theta,\te{r}) = l_2(\kappa)
  .
  \label{eq_subf_line_energy}
\end{equation}

\subsection{Stress Balance}\label{ss_derviationStress}
The stress balance \cref{eq_balance_stress} has a tensorial rank of two and, therefore, needs the corresponding 2-tensor test function \(\tee{\psi}\). Normalization with the factor \((1/2)\) yields
\begin{equation}
  \frac{2}{5} \int_\Omega {(\te{\nabla} \te{s})}_{\text{stf}} \tee{:} \tee{\psi} \dd \te{x} + \int_\Omega {(\te{\nabla} \te{u})}_{\text{stf}} \tee{:} \tee{\psi} \dd \te{x} + \frac{1}{2} \int_\Omega \left(\te{\nabla} \te{\cdot} \teee{m} \right) \tee{:} \tee{\psi} \dd \te{x} + \frac{1}{2} \frac{1}{\Knud} \int_\Omega \tee{\sigma} \tee{:} \tee{\psi} \dd \te{x} = 0
  .
\end{equation}
We reconsider that the stress tensor \(\tee{\sigma}\) is symmetric and trace-free and chose the same properties also for \(\tee{\psi}\), such that
\begin{equation}
  \tee{\sigma} =
  \begin{pmatrix}
    \sigma_{xx} & \sigma_{xy} & 0 \\
    \sigma_{xy} & \sigma_{yy} & 0 \\
    0           & 0           & -\left(\sigma_{xx} + \sigma_{yy}\right)
  \end{pmatrix},
  \quad
  \tee{\psi} =
  \begin{pmatrix}
    \psi_{xx} & \psi_{xy} & 0 \\
    \psi_{xy} & \psi_{yy} & 0 \\
    0           & 0           & -\left(\psi_{xx} + \psi_{yy}\right)
  \end{pmatrix}.
  \label{eq:stressTensorShapes}
\end{equation}
The symmetric and trace-free velocity gradient, therefore, expands to
\begin{equation}
  \int_\Omega {\left( \te{\nabla} \te{u} \right)}_{\text{stf}} \tee{:} \tee{\psi} \dd \te{x} = \int_\Omega {\left( \te{\nabla} \te{u} \right)}_{\mathrm{sym}} \tee{:} \tee{\psi} \dd \te{x} - \frac{1}{3} \int_\Omega \left( \te{\nabla} \te{\cdot} \te{u} \right) \trace({\tee{\psi}}) \dd \te{x},
\end{equation}
where the trace of \(\tee{\psi}\) directly vanishes due to the present setup. Using the additive tensor decomposition's orthogonality of the symmetric and skew-symmetric components, integration by parts becomes possible as
\begin{equation}\label{eq:stressIpU}
  \int_\Omega {\left( \te{\nabla} \te{u} \right)}_{\text{stf}} \tee{:} \tee{\psi} \dd \te{x}
  =
  \int_\Omega \te{\nabla} \te{u} \tee{:} \tee{\psi} \dd \te{x}
  =
  - \int_\Omega \te{u} \tee{\cdot} \left( \te{\nabla} \te{\cdot} \tee{\psi} \right) \dd \te{x} + \int_\Gamma \te{u} \te{\cdot} \left( \tee{\psi} \te{\cdot} \te{n} \right) \dd l
  .
\end{equation}
The integral with the divergence of the 3-tensor \(\teee{m}\) utilizes integration by parts, such that a Frobenius inner product of degree three and a boundary expression result as
\begin{equation}\label{eq:stressIp1}
  \int_\Omega (\nabla \cdot \teee{m}) \tee{:} \tee{\psi} \dd \te{x}
  =
  - \int_\Omega \teee{m} \teee{\because} \nabla \tee{\psi} \dd \te{x}
  + \int_{\Gamma} (\teee{m} \cdot \te{n}) : \tee{\psi} \dd l.
\end{equation}
To insert the boundary conditions, the terms of \cref{eq:stressIpU,eq:stressIp1} again use the local coordinate system with
\begin{align}
  (\teee{m} \cdot \te{n}) : \tee{\psi}
  &=
  m_{nnn}\psi_{nn} + 2 m_{nnt}\psi_{nt} + m_{ntt}\psi_{tt} + (-m_{ntt}-m_{nnn})(-\psi_{tt}-\psi_{nn})
  \nonumber
  \\
  &=
  \frac{3}{2} m_{nnn}\psi_{nn} + 2 m_{nnt}\psi_{nt} + 2\left(m_{ntt}+\frac{1}{2}m_{nnn}\right)\left(\psi_{tt}+\frac{1}{2}\psi_{nn}\right),
  \\
  \te{u} \cdot (\tee{\psi} \cdot \te{n})
  &=
  u_n \psi_{nn} + u_t \psi_{nt},
\end{align}
where \(m_{nnt} = m_{ntn}\) was utilized. All boundary terms can now be collected as
\begin{multline}
  \frac{1}{2} \int_{\Gamma} (\teee{m} \cdot \te{n}) : \tee{\psi} \dd l
  +
  \int_{\Gamma} \te{u} \cdot (\tee{\psi} \cdot \te{n}) \dd l
  =
  \int_{\Gamma} \left( \frac{3}{4} m_{nnn} + u_n \right) \psi_{nn} \dd l
  \\
  +
  \int_{\Gamma} \left( m_{nnt} + u_t \right) \psi_{nt} \dd l
  +
  \int_{\Gamma} \left(\frac{1}{2}m_{nnn}+m_{ntt}\right)\left(\frac{1}{2}\psi_{nn}+\psi_{tt}\right) \dd l,
\end{multline}
where the reordered boundary conditions \cref{eq_bcnew_un,eq_bcnew_sigmant,eq_bcnew_mnnn,eq_bcnew_05mnnnmnnt} fit naturally. Eliminating the highest-order moment \(\teee{m}\) with the closure relation \cref{eq_closure_m} yields the resulting equation as
\begin{align}
  &
  \int_\Gamma \left( u_t^{\text{w}} \psi_{nt} + \left( u_n^{\text{w}} - \epsilon^{\text{w}} \tilde{\chi} p^{\text{w}} \right)  \psi_{nn} \right) \dd l
  \nonumber
  \\
  &
  =
  \Knud \int_\Omega \text{stf}(\te{\nabla}\tee{\sigma}) \teee{\because} \text{stf}(\te{\nabla}\tee{\psi}) \dd \te{x}
  +
  \frac{1}{2} \frac{1}{\Knud} \int_\Omega \tee{\sigma} \tee{:} \tee{\psi} \dd \te{x}
  +
  \frac{9}{8} \tilde{\chi} \int_\Gamma \sigma_{nn} \psi_{nn} \dd l
  \nonumber
  \\
  &+
  \tilde{\chi} \int_\Gamma \left( \sigma_{tt} + \frac{1}{2} \sigma_{nn} \right) \left( \psi_{tt} + \frac{1}{2} \psi_{nn} \right) \dd l
  +
  \frac{1}{\tilde{\chi}} \int_\Gamma \sigma_{nt} \psi_{nt} \dd l
  +
  \epsilon^{\text{w}} \tilde{\chi} \int_\Gamma \sigma_{nn} \psi_{nn} \dd l
  \nonumber
  \\
  &+
  \frac{2}{5} \int_\Omega \tee{\psi} \tee{:} \te{\nabla} \te{s} \dd \te{x}
  - \frac{3}{20} \int_\Gamma \psi_{nn} s_n \dd l
  - \frac{1}{5} \int_\Gamma \psi_{nt} s_t \dd l
  - \int_\Omega \text{div}(\tee{\psi}) \te{\cdot} \te{u} \dd \te{x}
  + \epsilon^{\text{w}} \tilde{\chi} \int_\Gamma p \psi_{nn} \dd l
  ,
\end{align}
and using the sub-functionals defined in \cref{eq_subf_c,eq_subf_d,eq_subf_e,eq_subf_f,eq_subf_l3}, we have
\begin{equation}
  c(\te{s},\tee{\psi})
  +
  d(\tee{\sigma},\tee{\psi})
  -
  e(\te{u},\tee{\psi})
  +
  f(p,\tee{\psi})
  =
  l_3(\tee{\psi})
  .
  \label{eq_subf_line_stress}
\end{equation}

\subsection{Momentum Balance}\label{ss_derviationMomentum}
In contrast to the extensive derivation for the stress balance, we obtain the weak momentum balance trivially by multiplicating \cref{eq_balance_momentum} with \(\te{v}\). The absence of higher-order moments results in no need for integration by parts in
\begin{equation}
  \int_\Omega \te{\nabla} p \te{\cdot} \te{v} \dd \te{x} + \int_\Omega \left( \te{\nabla} \te{\cdot} \tee{\sigma} \right) \te{\cdot} \te{v} \dd \te{x} = \int_\Omega \te{b} \te{\cdot} \te{v} \dd \te{x}
  ,
\end{equation}
such that the resulting weak form \textendash\ using the sub-functionals \cref{eq_subf_e,eq_subf_g,eq_subf_l4} \textendash\ reads
\begin{equation}
  e(\te{v},\tee{\sigma})
  +
  g(p,\te{v})
  =
  l_4(\te{v})
  .
  \label{eq_subf_line_momentum}
\end{equation}

\subsection{Mass Balance}\label{ss_derviationMass}
Application of integration by parts for the mass balance from \cref{eq_balance_mass} allows enforcing the velocity boundary condition for \(u_n\) as
\begin{align}
  \int_\Omega \left( \te{\nabla} \te{\cdot} \te{u} \right) q \dd \te{x}
  &=
  - \int_\Omega \te{u} \te{\cdot} \te{\nabla}q \dd \te{x} + \int_\Gamma u_n q \dd l
  \\
  &=
  - \int_\Omega \te{u} \te{\cdot} \te{\nabla}q \dd \te{x}
  + \int_\Gamma \left( \epsilon^\mathrm{w} \tilde{\chi} \left( (p-p^\mathrm{w}) + \sigma_{nn} \right) + u_n^{\mathrm{w}} \right) q \dd l
  .
\end{align}
A reordering for \(\te{\mathcal{U}}\)- and \(\te{\mathcal{V}}\)-variables yields the weak formulation as
\begin{equation}
  \int_\Omega \te{u} \te{\cdot} \te{\nabla}q \dd \te{x}
  \int_\Gamma \epsilon^\mathrm{w} \tilde{\chi} \left( p + \sigma_{nn}\right) q \dd l
  =
  \int_\Omega \dot{m} q \dd \te{x}
  -
  \int_\Gamma \left( u_n^{\mathrm{w}} - \epsilon^\mathrm{w} p^\mathrm{w} \right) q \dd l,
  \label{eq_weakmass}
\end{equation}
that has to hold for all \(q \in \mathbb{V}_p\). The equation \cref{eq_weakmass} reads, in terms of sub-functionals \cref{eq_subf_f,eq_subf_g,eq_subf_h,eq_subf_l5}, as
\begin{equation}
  f(q,\tee{\sigma})
  -
  g(q,\te{u})
  +
  h(p,q)
  =
  l_5(q)
  .
  \label{eq_subf_line_mass}
\end{equation}


\end{document}